\documentclass{article}

\usepackage{arxiv}

\usepackage[utf8]{inputenc} % allow utf-8 input
\usepackage[T1]{fontenc}    % use 8-bit T1 fonts
\usepackage{hyperref}       % hyperlinks
\usepackage{url}            % simple URL typesetting
\usepackage{booktabs}       % professional-quality tables
\usepackage{amsfonts}       % blackboard math symbols
\usepackage{nicefrac}       % compact symbols for 1/2, etc.
\usepackage{microtype}      % microtypography
\usepackage{lipsum}		% Can be removed after putting your text content
\usepackage{graphicx}
\usepackage{doi}

\title{Asteroid Mining: ACT\&Friends' Results for the GTOC 12 Problem}

\date{} 					% Or removing it

\author{ {Dario Izzo}\\
	European Space Technology and Research Centre\\
	Noordwijk 2201 AZ, Netherlands\\
	\texttt{dario.izzo@esa.int} \\
	\And
	{Marcus Märtens} \\
	European Space Technology and Research Centre\\
	Noordwijk 2201 AZ, Netherlands\\
	\And
	{Laurent Beauregard} \\
	European Space Operations Centre~~~~~\\
	Darmstadt, 64293, Germany\\
	\And
	{~~~Max Bannach} \\
	European Space Technology and Research Centre\\
	Noordwijk 2201 AZ, Netherlands\\
	\And
	{Giacomo Acciarini} \\
	Surrey Space Center, University of Surrey\\
	GU27XH, Guildford, United Kingdom\\
	\And
	{Emmanuel Blazquez} \\
	European Space Technology and Research Centre\\
	Noordwijk 2201 AZ, Netherlands\\
	\And
	{Alexander Hadjiivanov} \\
	European Space Technology and Research Centre\\
	Noordwijk 2201 AZ, Netherlands\\
	\And
	{Jai Grover} \\
	European Space Technology and Research Centre\\
	Noordwijk 2201 AZ, Netherlands\\
	\And
	{Gernot Heißel} \\
	European Space Technology and Research Centre\\
	Noordwijk 2201 AZ, Netherlands\\
	\And
	{Yuri Shimane} \\
	School of Aerospace Engineering\\
	Georgia Institute of Technology\\
	Atlanta, Georgia 30332, United States of America\\
	\And
	{Chit Hong Yam} \\
	Mission Design and Operations Group, ispace inc.\\
	Sumitomo Fudosan Hamacho Building 3F, 3-42-3\\
	Nihonbashi Hamacho, Chuo-ku, Tokyo, Japan 103-0007
}

\renewcommand{\shorttitle}{\textit{Asteroid Mining: ACT\&Friends' result for the GTOC 12 problem}}

\hypersetup{
pdftitle={Asteroid Mining: ACT\&Friends' Results for the GTOC 12 Problem},
pdfsubject={cs.AI, physics.space-ph},
pdfauthor={Dario Izzo, Marcus Märtens, Laurent Beauregard, Max Bannach, Giacomo Acciarini, Emmanuel Blazquez, Alexander Hadjiivanov, Jai Grover, Gernot Heißel, Yuri Shimane, Chit Hong Yam},
pdfkeywords={trajectory optimization, space, asteroid mining, machine learning},
}

\usepackage{xcolor}
\usepackage{rotating}
\usepackage{subcaption}
\usepackage{graphicx}
\usepackage{amsmath,amsthm,amssymb}
\def\vec{\mathbf}

\newtheorem{theorem}{Theorem}
\newtheorem{lemma}[theorem]{Lemma}

\newcommand\Class[1]{%
  \mathchoice%
      {\text{\normalfont\small$\mathrm{#1}$}}%
      {\text{\normalfont\small$\mathrm{#1}$}}%
      {\text{\normalfont$\mathrm{#1}$}}%
      {\text{\normalfont$\mathrm{#1}$}}%
}

\usepackage{siunitx}
\usepackage{tikz}
\usetikzlibrary{graphs, calc, arrows.meta, decorations.markings, shapes, positioning, backgrounds}
\definecolor{jade}{rgb}{0.03, 0.47, 0.19}
\definecolor{lightsalmon}{rgb}{1.0, 0.63, 0.48}

\definecolor{coast_color}{RGB}{68, 119, 170}
\definecolor{hop_color}{RGB}{68, 119, 170}
\definecolor{sma_color}{RGB}{102, 204, 238}
\definecolor{peri_color}{RGB}{34, 136, 51}
\definecolor{apo_color}{RGB}{170, 51, 119}

\newcommand\mbhist[2]{%
  \begin{tikzpicture}[
      xscale = 0.7,
      yscale = 0.75,
      axis/.style = {
        draw, semithick, ->, >={[round]Stealth}
      },
      bar/.style = {
        semithick,
        color = peri_color,
        fill  = peri_color!25
      }
    ]

    \begin{scope}[xscale=0.2, yscale = 0.1]

      \foreach \y in {10,20,...,80}{
        \draw[gray, thin] (0,\y) -- (50,\y);
        \node[anchor=east] at (0,\y) {\scriptsize\y};
      }
      
      \foreach \y/\l [count=\x] in {#1} {
        \ifnum\y>0
        \draw[bar] (\x-0.3,0) rectangle (\x+0.3,\y);
        \fi
        \ifnum\numexpr\x/4*4\relax=\x
        \node[anchor=east, rotate=90] at (\x, -.25) {\scriptsize\l};
        \fi
      }
      
      \draw[axis] (0,0) to (0,90);
      \draw[axis] (0,0) to (51,0);

      \node at (25,-12) {\small Collected Mass of the Ships in kg};

      \ifx\empty#2\else
      \node[rotate=90] at (-5,45) {\small Number of Ships with this Mass};
      \fi
    \end{scope}
  \end{tikzpicture}%
}

\newcommand\coloneq{\mathrel{\raise.4pt\hbox{:}{=}}}
\newcommand\eqcolon{\mathrel{{=}\raise.4pt\hbox{:}}}

\usepackage{listings}
\lstdefinelanguage{pseudocode}{
  morekeywords={
    algorithm,method,new,and,
    if,then,else,while,do,repeat,until,seq,
    seqdo,return,call,
    for,pardo,foreach,print,output,input,exit,
    break,loop,end,begin,goto,par,global,local,
    read,write,stop,idle,procedure,function,
    throw,catch
  },
  sensitive=true,
  morecomment=[l]{//},
  morestring=[b]",
  morestring=[s]{``}{''},
}
\lstdefinestyle{pseudocode}{
  language=pseudocode,
  basicstyle=\small\rmfamily,
  commentstyle=\upshape\color{black!50},
  keywordstyle=\bfseries\itshape,
  identifierstyle=\itshape,
  stringstyle=\rmfamily,
  columns=fullflexible,
  mathescape,
  literate={<-}{{$\gets$\ }}2,
  numbers=left,
  numberstyle=\scriptsize\sffamily,
}

\newcommand\commented[1]{\textcolor{black!50}{#1}}

\newcommand{\arrow}{\overrightarrow}
\newcommand{\worra}{\overleftarrow}

\begin{document}
\maketitle

\begin{abstract}
In 2023, the 12th edition of Global Trajectory Competition was organized around the problem referred to as \lq\lq Sustainable Asteroid Mining\rq\rq. This paper reports the developments that led to the solution proposed by ESA's Advanced Concepts Team. Beyond the fact that the proposed approach failed to rank higher than fourth in the final competition leader-board, several innovative fundamental methodologies were developed which have a broader application.
In particular, new methods based on machine learning as well as on manipulating the fundamental laws of astrodynamics were developed and able to fill with remarkable accuracy the gap between full low-thrust trajectories and their representation as impulsive Lambert transfers.
A novel technique was devised to formulate the challenge of optimal subset selection from a repository of pre-existing optimal mining trajectories as an integer linear programming problem.
Finally, the fundamental problem of searching for single optimal mining trajectories (mining and collecting all resources), albeit ignoring the possibility of having intra-ship collaboration and thus sub-optimal in the case of the GTOC12 problem, was efficiently solved by means of a novel search based on a look-ahead score and thus making sure to select asteroids that had chances to be re-visited later on.
\end{abstract}

% keywords can be removed

% making tables look nicer
\renewcommand{\arraystretch}{1.2}

\section*{Nomenclature}

\begin{center}
\begin{tabular}{|p{.15\linewidth}|p{.75\linewidth}|}
\hline
$\mathbf r$ & spacecraft position vector ($m$) \\\hline

$\mathbf v$ & spacecraft velocity vector ($m/s$) \\\hline

$m$ & mass ($kg$) \\\hline

$\mathbf x = [\mathbf r, \mathbf v]$ & spacecraft state, containing position and velocity. \\\hline

$\mathbf y = [\mathbf x, m]$ & spacecraft state, containing position, velocity, and mass. \\\hline

$t$ & epochs (specific points in time) \\\hline

$T_{max}$ & maximum thrust ($N$) \\\hline

$T$ & time of flight along an asteroid transfer ($s$) \\\hline

$I_{sp}$ & specific impulse ($s$) \\\hline

$\Delta v$ &  impulsive velocity changes ($m/s$) \\\hline

$\mathcal S$ &  a collection of spacecraft (or ships)\\\hline

$\sigma$ &  a spacecraft (or ship) \\\hline

$\mathcal A$ &  the collection of (GTOC12) asteroids \\\hline

$\alpha$ &  an asteroid \\\hline

$a$ & acceleration ($m / s^2$) \\\hline

$\mathbf v, v$ & boldface symbols indicate vectors, the non boldface version indicates their norm. \\\hline

\end{tabular}
\end{center}

\section{Introduction}
Global Trajectory Optimization Competitions (GTOC) \cite{gtoc1} represent a biennial cornerstone event within the international aerospace community, dedicated to addressing the intricacies of interplanetary trajectory optimization. The 12th edition of this well established competition, held in June-July 2023, proposed a challenging design of a \lq\lq sustainable asteroid mining\rq\rq\ mission. 
The problem demanded the concurrent extraction of resources from a set $\mathcal A$ of 60,000 target asteroids, to be accomplished during a fixed 15 years wide window (from 2035-Jan-01 to 2050-Jan-01) by multiple spacecraft. 
The participating spacecraft, dispatched from Earth and possibly flying by Venus and Mars, had to be meticulously designed to maximize the quantity of mined material returned to our home planet. 
A comprehensive exposition of the mathematical intricacies underpinning the problem definition can be found in \cite{gtoc12}, while in this paper we will primarily provide essential definitions and selectively reference these mathematical foundations. For the purpose of clarity, we shall employ the term \lq ship\rq\ interchangeably with \lq spacecraft.\rq\
In the context of the multi-spacecraft asteroid mining mission presented in GTOC12, each ship possesses the capability to deploy a specified number of mining devices onto the asteroids' surface. Furthermore, these ships have the capacity to collect mined resources if a mining device is already in place on the visited asteroid. 
Importantly, each ship is not confined to gathering material exclusively from asteroids where it initially deposited a miner; it can collect resources from asteroids where miners were deployed by other ships. 
We introduce the concept of $n$-ship \emph{ensemble}, defined by a set $\mathcal S := \left\{\sigma_i\right\}$ comprising ships capable of collectively visiting each asteroid precisely twice: once for miner deployment and once for resource collection. 
To uniquely indicate the asteroid visited by each ship, we represent it with an integer corresponding to the asteroid's identification, that is, we consider the set of asteroids given in the competition as $\mathcal A\subseteq\mathbb{N}$, but will also silently identify an $\alpha\in\mathcal A$ with the corresponding celestial object and all its properties.
As an illustrative example, a three-ship \emph{ensemble} might be represented as follows:
$$
\mathcal S_3 = \left\{
\begin{array}{l}
\sigma_1 = [12, 45, 33, 23, 45, 87] \\
\sigma_2 = [49, 9, 12, 84, 33, 23] \\
\sigma_3 = [13, 87, 49, 9, 84, 13]
\end{array}
\right.
$$
while an example of a one-ship \emph{ensemble} (which we will refer to also as a self-sufficient ship) is:
$$
\mathcal S_1 = \left\{
\begin{array}{l}
\sigma_1 = [12, 45, 33, 97, 23, 23, 33, 97, 12 ,45] \\
\end{array}
\right.
$$
While the cooperative aspect of this mining problem (i.e. using high dimensional \emph{ensembles}) is most interesting and eventually allows for mining and returning to Earth more mass, it also adds a great deal of complexity to the problem and relatively good missions can be obtained also using only self-sufficient ships. It is at this point necessary to mention that all the top three teams in the GTOC 12 leaderboard made extensive use of large \emph{ensembles} thus proving their worth and advantages. In this paper, instead, we will present the design of a mining mission that only considers self-sufficient ships (i.e. one-ship \emph{ensembles}) as our attempts to leverage the collaborative aspect of the mining mission were never able to match the quality of the solution we found considering only the simpler case.
The rest of the paper is structured as follows. In Section \S \ref{sec:overall} we outline the overall strategy to design such a mission. 
In Section \S \ref{sec:database} we discuss the creation of low-thrust trajectory databases for the initial and last mission phases, i.e., coming from and returning to the Earth. 
In the following Section \S \ref{sec:lt2impulse} we discuss the differences between low-thrust trajectory \emph{hops} (i.e., transfers between asteroids) and their impulsive representation as Lambert's arcs and present several innovative techniques developed to cancel such a gap. 
In Section \S \ref{sec:phoenix} we discuss the search algorithm used to assemble, independently, self-sufficient ships.
Finally, in Section \S \ref{sec:patch} we discuss the construction of a multi-spacecraft mission collectively satisfying the imposed constraints and we conclude, in Section \S \ref{sec:final_assembled_solution}, describing our final design.

\section{The Overall Strategy}
\label{sec:overall}
Each ship, as detailed in the problem description, has a starting mass $m_0 \le 3000$ kg, a dry mass $m_d = 500$ kg and can deliver a maximum thrust $T_{max} = 0.6$ N with a specific impulse $I_{sp} = 4000$ s. Each ship also carries $I$ miners each with a mass $m_M = 40$ kg.
Assuming an average ship with a propellant mass of $m_p=2000$ kg, the straightforward application of the Tsiolkovsky rocket equation $\Delta V = -I_{sp} g_0 \log\frac{m_f}{m_i}$ returns, for such a ship, a maximum $\Delta V$ capability of roughly $43$ km/s which is delivered, thrusting constantly at the maximum possible level, in around $t_{\mathrm{thrust}} = I_{sp} g_0 \frac{m_f-m_i}{T_{\mathrm{max}}} = 4.14$ years. 
It is thus evident that the available ships are very agile and capable of hopping between multiple asteroids. 
This is true also in view of the actual spatial density of the asteroids, reported in Figure \ref{fig:asteroids_3d_density}, which suggests most of the good transfer opportunities between asteroids to be in the bracket [2.7, 2.9] AU and hence the ships to be mostly thrusting in that area where the Sun gravity is relatively weak. 
A mining mission may be thus divided into three phases which we designed separately and then connected:
\begin{enumerate}
    \item Leaving the Earth: potentially utilizing Mars and Venus fly-bys, each spacecraft reaches a high-density region in the asteroid belt to maximize its opportunities to find favorable inter-asteroid transfers. Ideally, this trajectory should be designed early within the 15-year operational window to facilitate the commencement of mining operations and accumulate the largest possible quantity of material (with a defined mining rate of 10 kg/yr \cite{gtoc12}).
    \item Deploy and Mine: for spacecraft ensembles, an extensive search for sequences of low-thrust transfers between asteroids is performed to link the initial and final trajectory segments.
    \item Returning to Earth: initiated as late as possible within the 15-year operational window and also potentially employing Mars and Venus fly-bys, return trajectories from the asteroids are crafted to allow the return of the maximum amount of mined resources.
\end{enumerate}
In the process of developing our solutions during the competition, a significant portion of our achievements drew upon prior research. We provide a brief overview of these foundational elements.
To devise the low-thrust trajectories and address optimal control challenges, our approach leaned on the integration of a direct approach, as outlined by Sims \cite{sims1997preliminary}, in conjunction with the utilization of the Non-Linear Programming solver SNOPT \cite{gill2005snopt}. This combination proved instrumental in addressing efficiently the needed trajectory optimization problems.

For the task of asteroid selection, a key component in ranking potential target asteroids, we employed the orbital metric approach \cite{hennes2016fast}. This method, extensively tested and validated, was found to yield results consistent with the selections made by the winning trajectory in the GTOC11 competition \cite{zhang2023gtoc}.
In our efforts to construct machine-learning models for the solutions to the Optimal Control Problems, we relied strongly on the PyTorch library \cite{paszke2019pytorch}. PyTorch provided a robust framework for developing and training these models, facilitating the integration of deep learning techniques into our approach.
The PyKEP and PyGMO software packages were also indispensable for the implementation of fundamental astrodynamical routines and optimization methods \cite{izzo2012pygmo}. These software tools streamlined the development of astrodynamics-related components within our solutions.
Additionally, in solving the integer linear programming problems (ILPs) that arose during the asteroid selection phase, we leveraged the capabilities of the open-source SCIP optimization suite \cite{BestuzhevaEA23}. The combination of these tools and libraries enabled us to develop robust, data-driven solutions for  various optimal control problems we defined (see next section) and other critical aspects of the competition.

% spatial asteroid density plot
\begin{figure}[htp]
    \centering
    \includegraphics[width=0.9\linewidth]{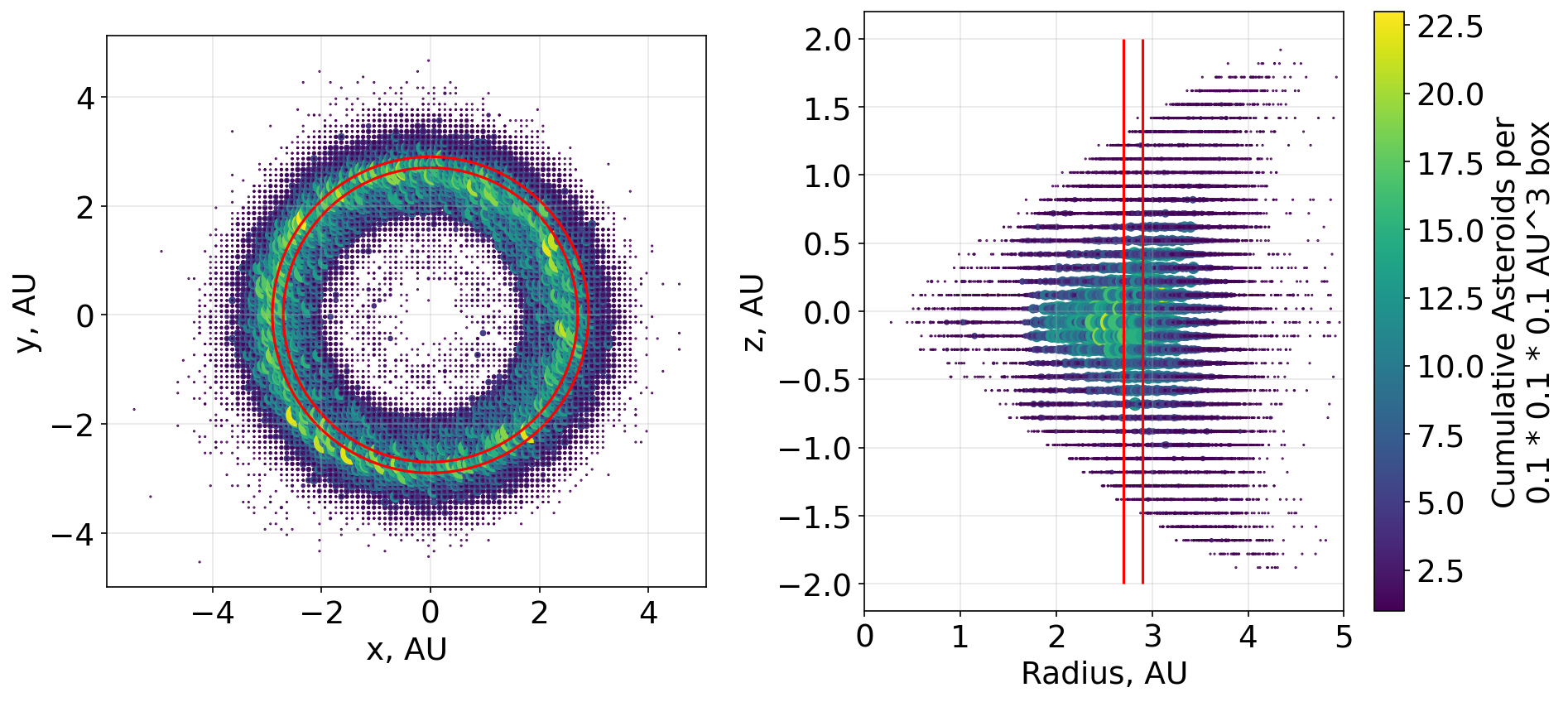}
    \caption{Spatial density of GTOC12 asteroids $x$-$y$ plane (left) and in the radius-$z$ plane (right), calculated by computing the asteroid positions at epochs 69807, 70107, and 70407 MJD corresponding to early mid and late mission epochs, and binning the Cartesian space with cubes of  $0.1$ AU size. Red lines show circular bounds between 2.7 and 2.9 AU.}
    \label{fig:asteroids_3d_density}
\end{figure}

\section{Databases with Departure and Return Options}
\label{sec:database}

%The first phase of the interplanetary trajectory of each ship, known as the departure leg, is a low-thrust transfer starting early from Earth and terminating at one of the asteroids, either via a direct transfer or by utilizing a Mars/Venus fly-by. Likewise, the final phase, referred to as the return leg, invariably involves departing from one of the asteroids, possibly late within the 15-year timeframe, to return collected material to Earth, also potentially facilitated by a Mars/Venus fly-by. 

A large number of candidate trajectories to bring the spacecraft from the Earth to their first asteroid, as well as the final leg departing the last visited asteroid to return to Earth, were computed and stored as databases to be used in the subsequent search of the \lq\lq deploy and mine\rq\rq\ phase.

Fly-bys with Venus or Mars have been under thorough consideration for both the initial and final legs of the trajectory. Extensive analysis has revealed that employing a Venus fly-by never offers any significant advantage over a direct transfer, primarily due to the fact that the majority of the GTOC12 asteroids are positioned beyond Earth's orbit. The use of a Mars fly-by, instead, may offer a substantial benefit, albeit only in rare cases. This is, for example, shown in Figure \ref{fig:db_scatter_arrival} where some optimal low-thrust transfers appear that use Mars fly-bys, but are outnumbered by similar opportunities not using Mars.

Overall, the relatively high departure and return $V_{\infty}$ allowed (6 km/s) from and to Earth, makes direct transfers particularly beneficial. Furthermore, the utilization of a fly-by generally necessitates precise alignment and phasing of the planets involved and the target asteroid: given the GTOC12 mission's requirements, we have found gravity assists to be mostly detrimental or marginally helpful.

% I think we should call "arrival" dbs as "First-leg" dbs because "arrival" is a bit ambiguous (YEPS changed to Departure)
\subsection{A Database of Departure Legs}
%Consider a low-thrust transfer leaving the Earth at epoch $t_0$ with a spacecraft mass $m_0$ and arriving at a destination asteroid at an epoch $t_f$ and with a spacecraft mass $m_f$. 

For each asteroid $\alpha$, we introduce and solve two distinct classes of optimal control problems, briefly described as follows:
\begin{equation} 
\mathcal P_1:
\left\{
    \begin{aligned}
        \mbox{min:} & \quad t_f
        \\
        \mbox{subject to:} &\quad
        m_f \geq m_{f\min}
        \\
        &\quad t_s \geq 64328 \quad MJD \\
        &\quad ...
    \end{aligned}
    \right.
\qquad
\mathcal P_2:
\left\{
       \begin{aligned}
        \mbox{max:} & \quad m_f
        \\
        \mbox{subject to:} &\quad
        t_f \leq t_{f\max}
        \\ &\quad
        t_s \geq 64328 \quad MJD \\
        &\quad ...
    \end{aligned} 
        \right.
\end{equation}
where, $m_{f\min} = [2300, 2500, 2600]$,  $t_{f\max} = 64328 + [570, 600, 650, 750]$ and direct, Venus, and Mars fly-by trajectory models are considered. The two optimal control problems classes reflect the desire to have both time optimal trajectories that are not using too much mass, and mass optimal trajectories not using too much time. The specific values used for $m_{f\min}$  and $t_{f\max}$ were tuned and finally chosen as we performed experiments during the competition. The modified Julian date of $64328$ corresponds to the Gregorian date 2035-Jan-01, i.e. the start of the mission window. We also use the subscripts $s,f$ to indicate the start and final conditions along the transfer and omit a formal description of all the constraints resulting from the dynamics, the boundary conditions, and the control constraints as they are not important in this context.
We thus created $60,000 \times 7$ optimal control problems that were solved using a simplified direct transcription method based loosely on the well-established Sims-Flanagan approach \cite{sims1997preliminary}. For transfers involving fly-by at Venus or Mars, a pre-screening based on a single deep space manoeuvre (DSM) model \cite{vasile2006preliminary} was conducted to prune for candidate asteroids that are well-phased with the planet in question. 

% plot of first leg db
\begin{figure}[t]
    \centering
    \includegraphics[width=0.9\linewidth]{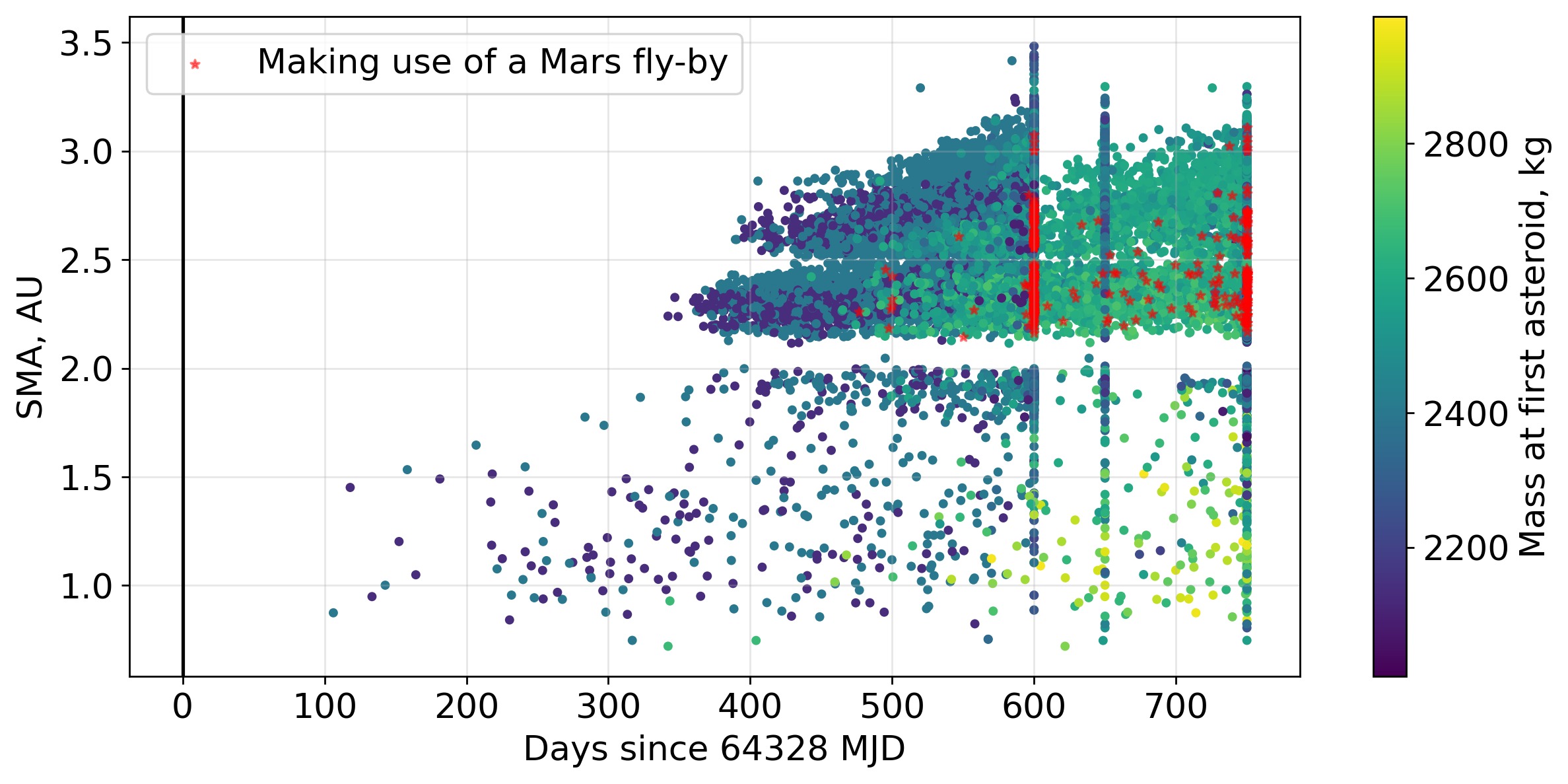}
    \caption{Visualization of the low-thrust transfers included in the departure database~$\mathcal D_{\text{dep}}$. A total of 91,135 non-dominated transfers to 18,504 asteroids were included. The 665 transfers making use of a Mars fly-by are marked. The transfers hitting the bounds on $t_{f\max}$ are clearly visible as vertical lines.}
    \label{fig:db_scatter_arrival}
\end{figure}

The result of this departure leg search indicated that among the 60,000 asteroids available, only a subset of approximately 18,000 exhibited a favorable phasing and were consequently considered as potential initial targets.
All the most promising optimal transfer options (i.e. non-dominated) were stored in a database $\mathcal D_{\text{dep}}\subseteq\mathcal A\times\mathbb R\times\mathbb R$ that contains an element $(\alpha,t,m)$ if there is a departure leg that arrives at asteroid $\alpha$ at time $t$ with a remaining mass of $m$. Figure~\ref{fig:db_scatter_arrival} 
visualizes $\mathcal D_{\text{dep}}$. Only 665 transfers in this non-dominated set made use of a Mars fly-by, and none used Venus.

\subsection{A Database of Return Legs}
As to compute a number of possible return options, for each asteroid $\alpha$, we introduce the following optimal control problem:
\begin{equation} 
\mathcal P_3:
\left\{
    \begin{aligned}
        \mbox{max:} & \quad m_f
        \\
        \mbox{subject to:} 
        & \quad t_s \geq 69807 - \Delta t_{\max} \\
        & \quad t_f \leq 69807 \\
        & \quad ...
    \end{aligned}
    \right.
\end{equation}
The modified Julian date $69807$ corresponds to the Gregorian date 2050-Jan-01, i.e., the end of the mission window. 
In essence, we compute the maximum amount of mined material that can be returned to Earth from each asteroid assuming some initial mass $m_s$ and to depart not earlier than $\Delta t_{\max}$ days from the end of the mission window. 
For favorable return options, this value is expected to decrease, potentially in a nonlinear fashion, as $\Delta t_{\max}$ increases, signifying a gradual loss of asteroid-Earth phasing.
We have solved this optimal control problem for a range of $\Delta t_{\max}$ values, specifically $\Delta t_{\max} = {800, 750, 700, 650, 550, 500}$ days.
Although we initially considered various departure mass values $m_s$, our preliminary findings consistently indicated that the ratio $m_f/m_s$ (i.e., the required $\Delta v$) remained relatively constant for a given asteroid. Consequently, we opted to focus exclusively on one value, namely $m_s=1200$ kg, for the construction of the return database. Formally, the database $\mathcal D_{\text{ret}}\subseteq\mathcal A\times\mathbb R\times\mathbb R$ contains, as the departure database, triplets $(\alpha,t,m)$ indicated that a spacecraft of mass $m$ on an asteroid $\alpha$ at time $t$ can return to Earth (no gravity assists options were found that helped returning more mass to the Earth.). The entries within the return databases are depicted in Figure~\ref{fig:db_scatter_return}.

% plot of return db
\begin{figure}[t]
    \centering
    \includegraphics[width=0.9\linewidth]{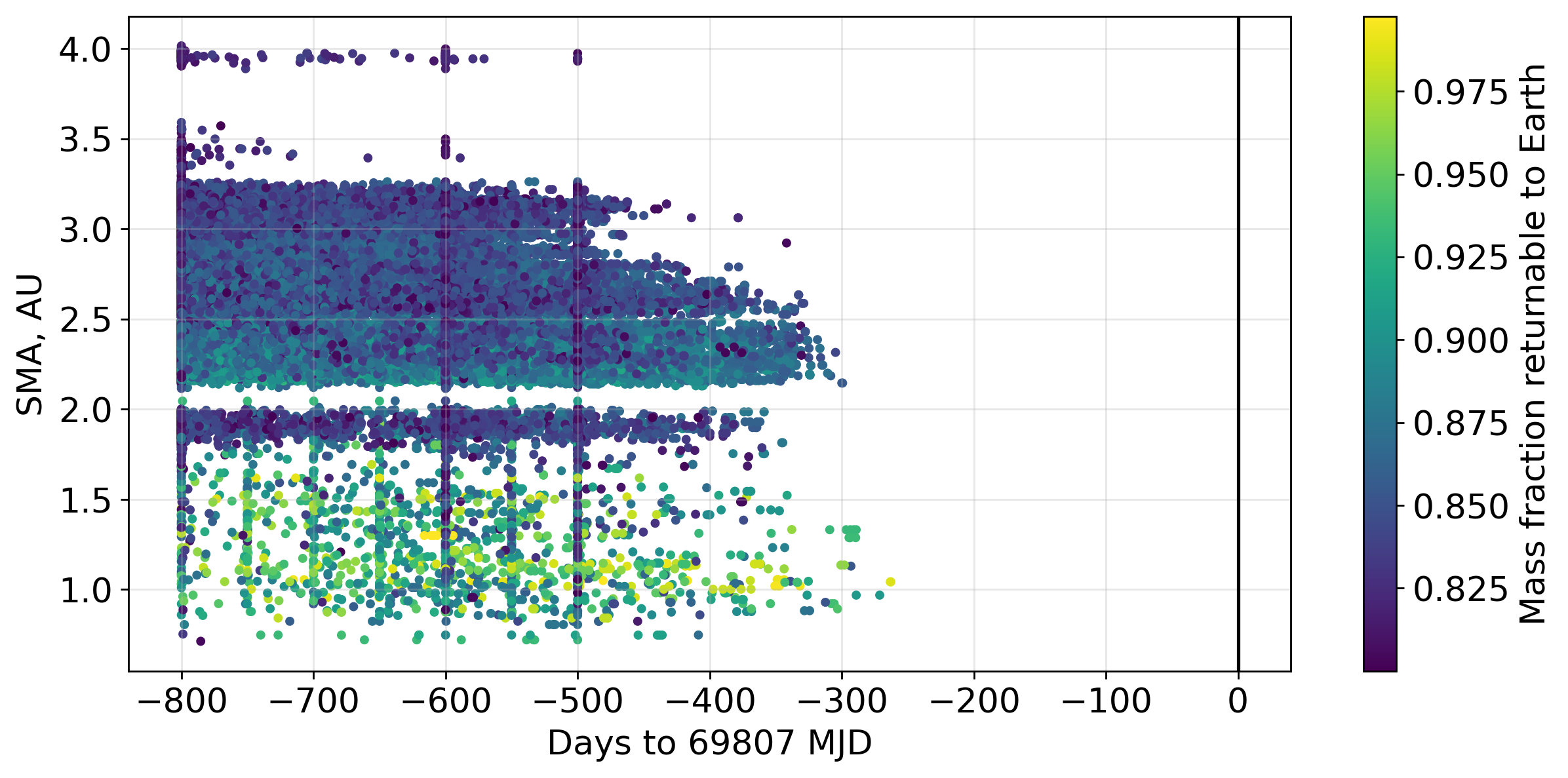}
    \caption{Departure epoch vs semi-major axes for all entries in the return database~$D_{\text{ret}}$, consisting of 190,976 transfers to 59,993 asteroids.}
    \label{fig:db_scatter_return}
\end{figure}

\section{Filling the Gap Between Impulsive and Low-thrust Solutions}
\label{sec:lt2impulse}

%Given a rendezvous problem, where a spacecraft equipped with low-thrust propulsion aims to hop from a certain asteroid at a departure epoch, to another one, it is essential to understand:
%\begin{itemize}
%    \item if the transfer between the two asteroids is feasible at all. Does the %initial mass of the spacecraft enable the transfer, provided a fixed time of %%flight?
%    \item even if the transfer is possible: what is the total cost of the %transfer that minimizes the fuel?
%\end{itemize}
To allow for the efficient design of multiple asteroid rendezvous missions utilizing a purely impulsive model, it is imperative to comprehend the difference with its low-thrust counterpart.
Using the impulsive $\Delta v$ to extract an estimation of the actual propellant mass to be used is a widely employed method, but also one well-known to lead to sub-optimal designs. On many occasions, a given transfer might not be feasible at all in low-thrust, since the spacecraft's initial mass, or the desired time of flight, might be too constraining to enable the transfer.
We thus developed several approximations to have a good understanding of whether a certain given spacecraft's initial mass is sufficient to make a fixed-time transfer. We then leveraged the Lambert solution and developed analytical and machine learning (ML)-based approximations that helped us to accurately transition from the high to the low-thrust world, without directly solving an optimal control problem.
%In Section~\ref{sec:mim_and_mint}, we introduce two fundamental quantities: the Maximum Initial Mass (\textsc{MIM}), defined for fixed-time low-thrust transfers, and the Minimum Time-of-flight (\textsc{MINT}), defined for fixed initial mass transfers. These are defined as solutions to their respective optimal control problems. Then, in Section ~\ref{sec:analytical_approximations}, we discuss analytical approximations for those quantities, as well as the minimum propellant required to perform the transfer. Finally, in Section~\ref{sec:machine_learning_approximations}, we introduce and discuss \textsc{ML}-based approximations.

\subsection{MIM and MINT: the Maximum Initial Mass and the Minimum Time-of-flight}
\label{sec:mim_and_mint}
Let us consider a rendezvous trajectory between two asteroids $\alpha_s$ and let $\alpha_f$, $t_s$, and $t_f$ be the starting and arrival epochs of the spacecraft at the two asteroids. 
In correspondence to such a fixed time of flight, there will be one well-defined transfer maximizing the initial spacecraft mass. 
This transfer corresponds to a trajectory with no coast arcs since if a coast arc was present in the optimal solution, it would be possible to have an infinitesimally larger initial mass, contradicting the optimality assumption. The optimal control problem defining the maximum initial mass $m^*$ can be formally stated as:

\begin{equation} 
\mathcal P_{\textsc{MIM}}:
\left\{
    \begin{aligned}
            \mbox{given:} & \quad \alpha_s, \alpha_f, t_s, t_f \\
            \mbox{find:} & \quad \mathbf u(t) \\
        \mbox{to maximize:} & \quad m_s
        \\
        \mbox{subject to:} &\quad
        \mathbf{x}(t_s)-\mathbf{x}_{\alpha_s}(t_s)=\mathbf{0}\\
        &\quad \mathbf{x}(t_f)-\mathbf{x}_{\alpha_f}(t_f)=\mathbf{0}
        \\
        &\quad \dot{\mathbf y} = \mathbf f(\mathbf y, \mathbf u)\\
        &\quad |\mathbf u| \le T_{\mathrm{max}}
    \end{aligned}
    \right.
    \text{,}
    \label{eq:optimal_control_problem_mim}
\end{equation}
where $\mathbf{x}$ indicates the state (e.g., a six-dimensional vector expressing the position and velocity in Cartesian coordinates),  $\mathbf{y} = [\mathbf x, m]^T$ indicates the state augmented with the mass, the subscripts $\alpha$ indicate the asteroids and $\mathbf f$ indicates the spacecraft controlled dynamics. 
If the initial spacecraft mass were to be higher than the maximum initial mass $m_s > m_s^*$, then the transfer, also referred to as a \textit{hop}, cannot be performed in low-thrust (i.e., the spacecraft is said to be \lq\lq fat\rq\rq\ as it weighs too much to be able to perform the hop). 
The optimal control problem shown in Eq.~\eqref{eq:optimal_control_problem_mim} can be solved by varying $t_f$: the resulting optimal starting mass $m_s$ will be a continuous curve starting from the origin since to perform a hop with an infinitesimal time of flight an infinitesimally small mass is needed. In Fig.~\ref{fig:mim_vs_mint}, we show an example of such a curve computed solving accurately the corresponding optimal control problems for six different transfers between asteroids of the GTOC12 database corresponding to randomly selected $t_s$ values. 
Given a point on the curve, when the spacecraft's initial mass is below that value, the spacecraft is said to be slim and the transfer is feasible, otherwise, the spacecraft is fat and the transfer is not feasible.
\begin{figure}[htp]
    \centering
    \includegraphics[width=0.85\linewidth]{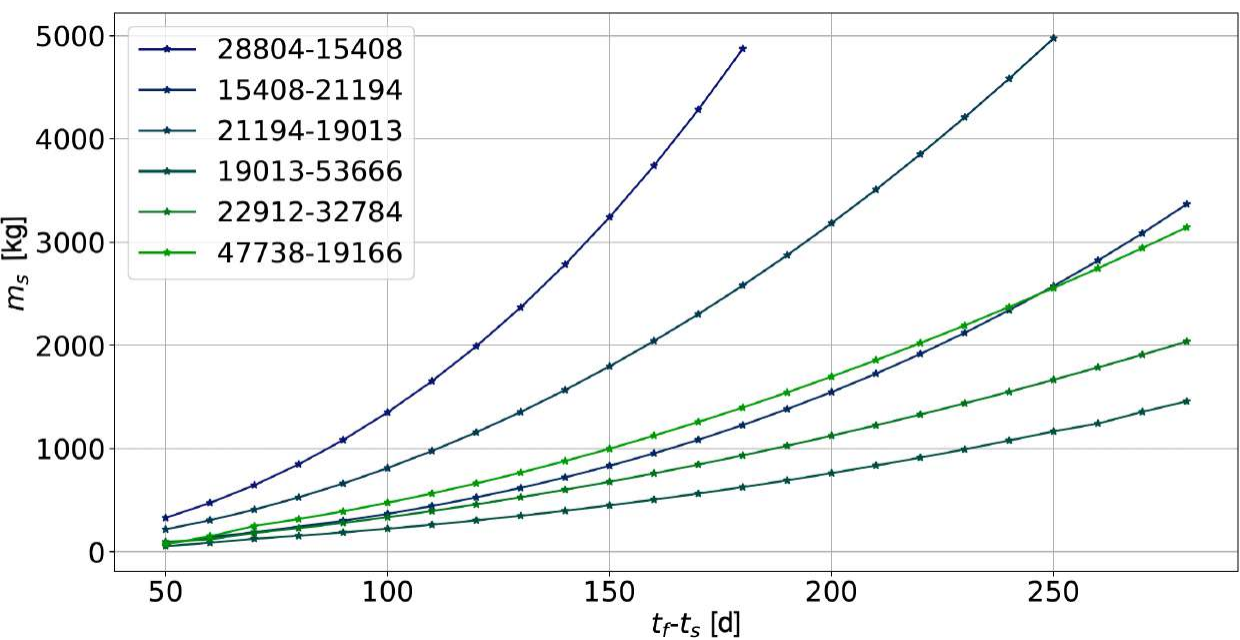}
    \caption{Maximum initial mass curve as a function of time from the start of the mission for different asteroid hops. In the legend, we display the database ID of the departing and arriving asteroids.}
    \label{fig:mim_vs_mint}
\end{figure}

One may now ask the dual question: given an assigned starting mass $m_s$, what is the minimum time of flight our spacecraft can rendezvous with the destination asteroid? The answer indicated with $t_f^*$ and called \textsc{MINT}, is mathematically the solution of the following optimal control problem:
\begin{equation} 
\mathcal P_{\textsc{MINT}}:
\left\{
    \begin{aligned}
            \mbox{given:} & \quad \alpha_s, \alpha_f, t_s, m_s \\
            \mbox{find:} & \quad \mathbf u(t) \\
        \mbox{to minimize:} & \quad t_f
        \\
        \mbox{subject to:} &\quad
        \mathbf{x}(t_s)-\mathbf{x}_{\alpha_s}(t_s)=\mathbf{0}\\
        &\quad \mathbf{x}(t_f)-\mathbf{x}_{\alpha_f}(t_f)=\mathbf{0}
        \\
        &\quad \dot{\mathbf y} = \mathbf f(\mathbf y, \mathbf u)\\
        &\quad |\mathbf u| \le T_{\mathrm{max}}
    \end{aligned}
    \right.
    \text{,}
    \label{eq:mint}
\end{equation}
To solve this problem, one may simply invert the curve $m_s^*(t_f)$ into $t_f^*(m_s)$, thus finding the intersection of a horizontal fixed mass line with the plot in Figure~\ref{fig:mim_vs_mint}. Computing the \textsc{MIM} and \textsc{MINT} from their definitions as solutions to optimal control problems is computationally demanding and is likely to be not useful. Hence, we propose a few fast approximators, considering that for the problem of interest, we can expect most thrust arcs to happen in less than a complete revolution. 
%In Section~\ref{sec:analytical_approximations}, we will detail the analytical approximations, while in Section~\ref{sec:machine_learning_approximations}, we discuss the \textsc{ML}-based ones.

\subsection{Analytical Approximations}
\label{sec:analytical_approximations}
Consider the low-thrust transfer from $\alpha_s$ at epoch $t_s$ to $\alpha_f$ at epoch $t_f$. As a naive approximation for the maximum initial mass $m^*$, we can compute a total $\Delta v$ solving the Lambert problem with transfer time $T$ and summing the two resulting velocity discrepancies. One can then set:
\begin{align}
    m^* \approx \frac{T_{\text{max}}T}{\Delta v}
    \text{,}
\end{align}
However, this approximation is not very accurate. %An important point to note is the boundary point between a feasible and infeasible trajectory must be continuously thrusting.
The maximum initial mass approximation (\textsc{MIMA}) was developed with the idea of improving the above formula with a more sensible approximation~\cite{hennes2016fast}. 
To derive it, the motion of the spacecraft traveling from $\alpha_s$ to $\alpha_f$  is described as a variation with respect to a corresponding Lambert transfer defined by the two velocity impulses. 
It is assumed that the spacecraft is subject to a piecewise constant acceleration (both in magnitude and direction): this configuration, at the basis of the original developments to derive the \textsc{MIMA} summarized here, is depicted in Figure~\ref{fig:mima}. 
By imposing that the two accelerations are equal in magnitude and that they need to provide a total velocity change that corresponds to the Lambert cost, the problem can be formulated and solved to obtain the expression for the thrust magnitudes and directions, as well as the switching time. 
This results in the following approximation for the maximum initial mass~\cite{hennes2016fast}:
\begin{equation}\label{eq:mima}
    m_{\textsc{MIMA}}^*=2\dfrac{T_{\mathrm{max}}}{a_D}\bigg(1+\textrm{exp}\bigg( \dfrac{-a_D T}{I_{sp}g_0} \bigg)\bigg)^{-1}
    \
    \text{,}
\end{equation}
where $T$ is the time of flight, and $a_D$ is the magnitude of the acceleration of each of the two segments. 

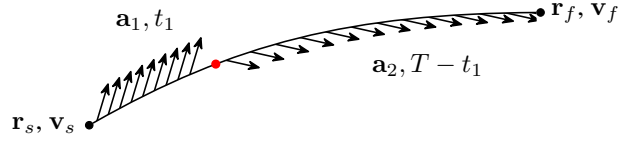
\begin{figure}[t]
    \centering
\begin{tikzpicture}[
    xscale = 1.5,
    trajectory/.style = {
      semithick
    },
    thrust/.style = {
      semithick,
      ->,
      >={[round]Stealth}
    },
    dot/.style = {
      draw, fill,
      circle,
      inner sep = 0pt,
      minimum width = 0.1cm
    },
    markeddot/.style = {
      dot,
      color = red
    }
  ]
  \path[use as bounding box] (-.75,-.5) rectangle (5,2);

  \node[dot] (s) at (0,0) {};
  \node[left of=s, node distance=0.6cm] {$\mathbf r_s$, $\mathbf v_s$};
  
  \node[dot] (t) at (4,1.5) {};
  \node[right of=t, node distance=0.6cm] {$\mathbf r_f$, $\mathbf v_f$};

  % line of top
  \newcounter{thrustnode}
  \draw[
    decoration={markings, mark=between positions 0.01 and 0.99 step 0.028
      with {\node[circle, draw=none, inner sep=0pt, minimum size=0pt] (thrustnode\thethrustnode) {};
        \stepcounter{thrustnode}}},
    postaction={decorate},
    trajectory
  ] (s) to[bend left=20] node[pos=0.3, markeddot] (x) {} (t);  

  % first half of thrust
  \begin{scope}
    \node at (0.5,1.4) {$\mathbf a_1, t_1$};
    \path[clip] (s) to[bend left=20] (t)
    -- ($(t)+(0,2)$) 
    -- ($(s)+(0,2)$)
    -- cycle; 
    \foreach \i in {0,1,2,...,8}{
      \draw[thrust] ($(thrustnode\i)-(0.1,0.5)$) to ($(thrustnode\i)+(0.1,0.5)$);
    }
  \end{scope}

  % seconf half of thrust
   \begin{scope}
    \node at (3,0.8) {$\mathbf a_2, T-t_1$};
    \path[clip] (s) to[bend left=20] (t)
    -- ($(t)-(0,2)$) 
    -- ($(s)-(0,2)$)
    -- cycle;  
    \foreach \i in {11,13,...,35}{
      \draw[thrust] ($(thrustnode\i)-(0.3,-0.1)$) to ($(thrustnode\i)+(0.3,-0.1)$);
    }
  \end{scope}

\end{tikzpicture}
    \caption{To approximate the maximum initial mass (MIM) we find an analytical representation for a low-thrust transfer that employs a piecewise constant thrust. All variables $\mathbf a_1, \mathbf a_2, t_1$ can be found expanding the solution around the corresponding Lambert problem (i.e. a ballistic transfer from $\mathbf r_s$ to $\mathbf r_f$ in $T$).}
    \label{fig:mima}
\end{figure}
%Since the thrust is still significant compared to the mass, we can expect most thrust arcs to happen in less than a complete revolution. 

While \textsc{MIMA} is quick to calculate and provides a much better estimation of the feasibility of the jump compared to the naive approach, it still represents an idealized and impractical transfer that does not accurately represent reality. We thus developed a second approximation (aptly named \textsc{MIMA2}) incorporating first-order information about the trajectory. The derivation starts with the same spirit as for \textsc{MIMA} and is summarized below.
The spacecraft state $\mathbf x_K=[\mathbf{r}_K,\mathbf{v}_K]^\mathsf{T}\in\mathbb R^6$ along a ballistic arc connecting the positions of two asteroids in the time $T$ is the solution to the (Keplerian) differential equation:
\begin{align}\label{eq:kepler}
\dot{\mathbf x}_K = \mathbf F(\mathbf x_K)
\end{align}
where $\mathbf F = [\mathbf{v}_K,-\mu\mathbf{r}_K/r_K^3]^\mathsf{T}$. 
Executing this ballistic transfer requires a ship to perform two impulsive maneuvers, denoted as $\Delta \mathbf v_s$ and $\Delta \mathbf v_f$, at the initiation and conclusion of the arc. The magnitudes of these maneuvers are readily calculated to align the departure and arrival orbital velocities with those along the ballistic arc, which is determined as a solution to the Lambert problem \cite{lambert}.

Consider the spacecraft dynamics under the effect of a finite thrust. The equations of motion become:
\begin{equation}\label{eq:finite_thrust}
    \dot{\mathbf{x}} = \mathbf F(\mathbf x)
    + \tilde{\mathbf{a}}(t)%+ \mathbf N\mathbf{a}(t)
\end{equation}
with $\tilde{\mathbf{a}}(t)=[\mathbf 0,\mathbf a(t)]^\mathsf{T}$ and where $\mathbf a(t)$ is the additional, low-thrust induced, acceleration. % and $\mathbf N = \left[\begin{smallmatrix}0_{3\times 3} \\ I_{3\times 3}\end{smallmatrix}\right]$.
We write the spacecraft motion, during the low-thrust transfer between two asteroids, as a perturbed Keplerian arc and thus write
$
    \mathbf x = \mathbf x_K + \delta \mathbf x.
$ 
Substituting this into Eq.~\eqref{eq:finite_thrust}, Taylor approximating $\mathbf F$ to first order around $\mathbf x_K$ and using Eq.~\eqref{eq:kepler} yields
\begin{equation}
\label{eq:eom_thrust}
    \delta  \dot{\mathbf x} = \mathbf F'(\mathbf x_K)\delta \mathbf x  + \tilde{\mathbf a}(t)
\end{equation}
describing the dynamics of the perturbation vector. The rendezvous conditions with $\alpha_s$ and $\alpha_f$ are enforced by imposing the boundary conditions
$
    \delta \mathbf x(0) = \delta \mathbf x_0 = [\mathbf 0,  -\Delta \mathbf v_s]^\mathsf{T}, 
    \delta \mathbf x(T) = \delta \mathbf x_T = [\mathbf 0, \Delta \mathbf v_f]^\mathsf{T}
$,
where here and in the following a subscript on $\delta \mathbf x$ denotes the evaluation of that quantity at a specific time.
The homogeneous equation (i.e. Eq.~\eqref{eq:eom_thrust} with $\tilde{\mathbf a}=\mathbf 0$) admits the general solution $\delta \mathbf x(t) = \mathbf M(t)\delta  \mathbf x_0$, where $\mathbf M(t)$ is the state transition matrix. We use the variation of parameters technique (see e.g. \cite{battin}) to solve the inhomogeneous equation, hence we consider $\mathbf x_0$ as a time-varying quantity, change its symbol to $\mathbf z(t)$ and write the solution as $\delta \mathbf x(t) =  \mathbf M(t)  \mathbf z(t)$. Substituting this expression back into Eq.\eqref{eq:eom_thrust}, after simplification, we obtain the differential equation for $\mathbf z(t)$:
\begin{align}
\label{eq:eom_y}
    \dot{\mathbf z} = \mathbf M^{-1}(t) \tilde{\mathbf{a}}(t)
\end{align}

Equation~\eqref{eq:eom_y} is an ordinary differential equation for which we can explicitly write the solution at the time $T$ as:
\begin{equation*}
    \delta \mathbf z_T = \delta \mathbf z_0  +\int_{0}^{T}\mathbf M^{-1}(s)\tilde{\mathbf{a}}(s)\ ds
\end{equation*}
which, accounting for the definition of $\mathbf z(t)$ and the identity $\mathbf M_0 = \mathbf I$, becomes:
\begin{align}
\label{eq:sol_an_integral}
    \delta \mathbf x_T = \mathbf M_T\delta \mathbf x_0  + \mathbf M_T\int_{0}^{T}\mathbf M^{-1}(s)\tilde{\mathbf{a}}(s)\ ds
\end{align}
%where we used the notation $\mathbf M_T = \mathbf M(T)$ and we took into account that: $\mathbf y_0 = \delta \mathbf p_0$.
Note how the only unknown in this equation is the acceleration profile $\tilde{\mathbf{a}}(t)$.

We now assume, similarly to what is done in the case of the \textsc{MIMA} derivation, that the low-thrust transfer is made of two constant thrust arcs $\mathbf a_1, \mathbf a_2$ of duration $t_1$ and $T-t_1$. Approximating the resulting two integrals on the right-hand side of Eq.~\eqref{eq:sol_an_integral} by Simpson's rule we obtain,
\begin{equation}
%\begin{split}
    \delta \mathbf x_T-\mathbf M_T\delta \mathbf x_0 = \frac{\mathbf M_T}{6}\left(\left(\mathbf M^{-1}_0+4\mathbf M^{-1}_{t_1/2}+\mathbf M^{-1}_{t_1}\right)\Delta \mathbf v_1^*+
    \left(\mathbf M^{-1}_{T-t_2}+4\mathbf M^{-1}_{T-t_2/2}+\mathbf M^{-1}_T\right) \Delta \mathbf v_2^*\right) \label{dvs_star}
%\end{split}
\end{equation}
where $\Delta \mathbf v_1^* = [\mathbf 0, \mathbf{a}_1 t_1]^\mathsf{T}$, $\Delta \mathbf v_2^* = [\mathbf 0, \mathbf{a}_2 (T-t_1)]^\mathsf{T}$ and where for later convenience we substitute $t_2$ for $T-t_1$ in the subscripts of the second Simpson approximation.

Given a choice for the switching time $t_1$, the above system of equations admits only one solution in the unknowns $\Delta \mathbf v_1^*$ and $\Delta \mathbf v_2^*$ which only appear linearly. This corresponds to the fact that there is only a unique solution to satisfy the boundary conditions when we force the thrust structure into this two-segment piecewise-constant thrust law. We choose the switching time $t_1$ by demanding that the magnitude of the two accelerations be equal and thus solving the equation
\begin{equation}
(T-t_1) |\Delta \mathbf v^*_1| = t_1 |\Delta \mathbf v^*_2|
\end{equation}
iteratively for $t_1$.
We have thus found an approximation for a continuous thrust trajectory that is able to transfer between two generic asteroids in the time $T$ using two constant-thrust segments. An approximation that can be computed at the cost of evaluating the state transition matrix $\mathbf M$ at a few points along the ballistic transfer!

Such a trajectory is a good approximation for the optimal low-thrust transfer that typically results in solving the corresponding optimal control problem for the maximum initial mass. We thus compute the necessary acceleration from
\begin{align}
    a = |\mathbf a_1| = |\mathbf a_2| = \frac{|\Delta \mathbf v_1^*|}{t_1} = \frac{|\Delta \mathbf v_2^*|}{T-t_1}
\end{align}
and approximate the maximum initial mass via the equation:
\begin{equation}
m^*_{\textsc{MIMA2}} = \frac{T_{\mathrm{max}}}{a}
\text{.}    
\end{equation}

In Figure~\ref{fig:mima_mima2_errors}, we display the error expressed as the difference between the maximum initial mass of the solution of the optimal control problem and the approximated one, using either \textsc{MIMA} or \textsc{MIMA2} over around 400,000 low-thrust transfers with maximum initial masses within 700 and 3,000 kg. As we observe, \textsc{MIMA2} has usually smaller errors than \textsc{MIMA}, although its average errors are slightly skewed towards the positive x-axis. Moreover, about 58.34\% of the \textsc{MIMA2} errors fall within the [-50,50] kg range, while this number drops to 53.00\% for the \textsc{MIMA}.
\begin{figure}[htp]
\centering    \includegraphics[width=0.7\linewidth]{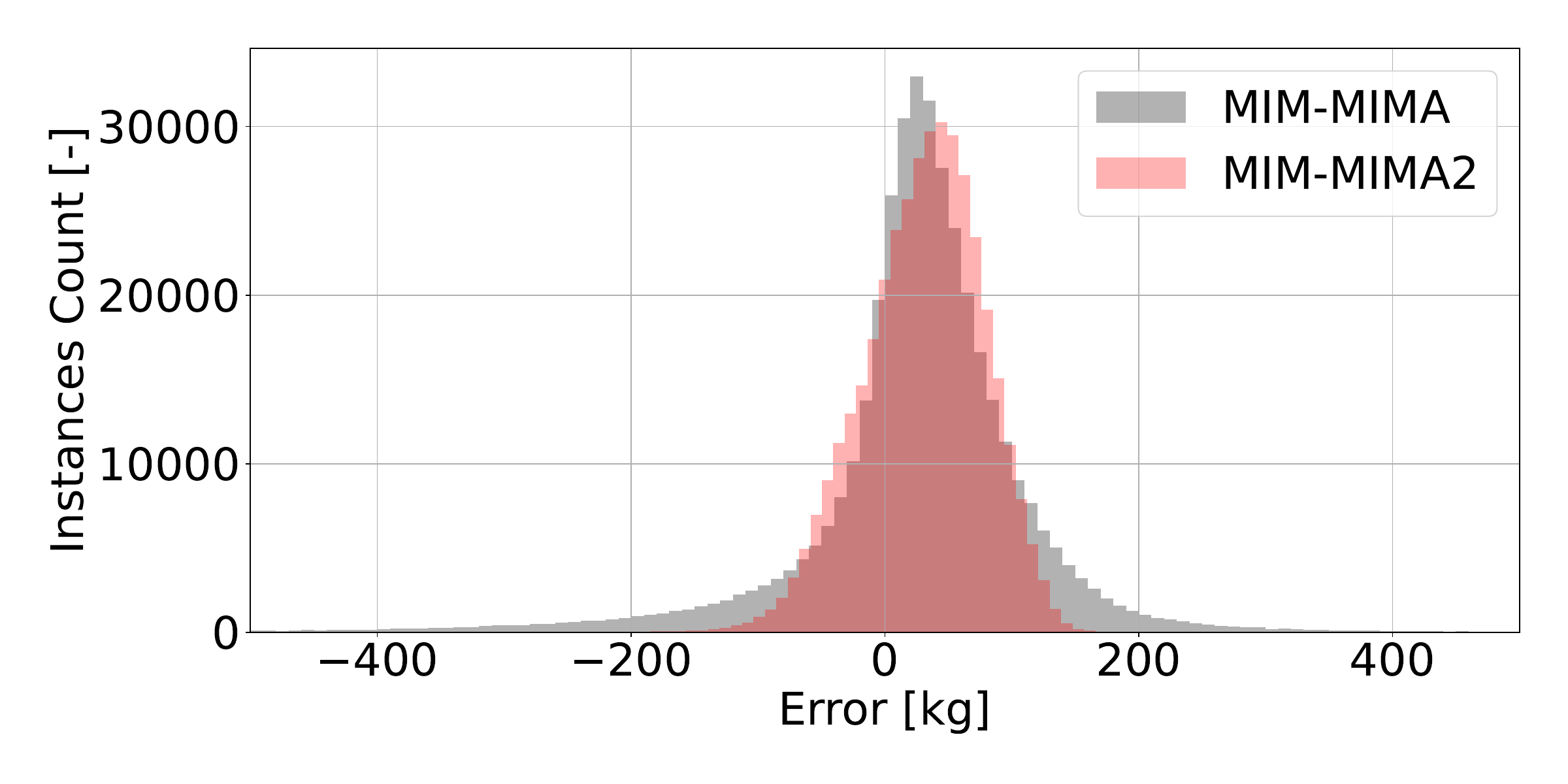}
    \caption{Difference in kilograms between \textsc{MIM} and \textsc{MIMA}, \textsc{MIMA2} approximations.}
    \label{fig:mima_mima2_errors}
\end{figure}

The total required $\Delta v$ can also be approximated for the MIM-MINT transfer using
$%\begin{align}
    \Delta v = |\Delta \mathbf v_1^*|+|\Delta \mathbf v_2^*|\label{dv_total}.
$ %\end{align}
In cases where a spacecraft has a mass $m_s$ lower than $m^*_{\textsc{MIMA2}}$, and optimally transfers in $T$ between the two asteroids, there will be a coasting arc and a lower $\Delta v$ requirement. In that situation Eq.(\ref{dvs_star}) is still valid, but the durations of the thrust arks, $t_1$ and $t_2$, are both unknown as they no longer sum up to $T$. Since we are assuming to know the starting mass, though, we may also assume $a = \frac{T_{\mathrm{max}}}{m_s}$. Thus, we solve the system $\{|\Delta \mathbf v_1^*| = t_1a,\, |\Delta \mathbf v_2^*| = t_2a\}$ for $t_1$ and $t_2$, and approximate the total $\Delta v$ needed for a generic hop, i.e. one that also has a coast arc.

As already discussed in Section~\ref{sec:mim_and_mint}, the maximum initial mass and minimum time of flight are closely related concepts. Using a root solver, one can derive from \textsc{MIMA} and \textsc{MIMA2}, the corresponding approximated minimum time of flights \textsc{MINTA} and \textsc{MINTA2}.

\subsection{Machine Learning Approximations}
\label{sec:machine_learning_approximations}
During the competition, a vast amount of optimal control problems of rendezvous transfers among pairs of asteroids were solved (e.g., to check the validity of the provided approximations). Therefore, it quickly became apparent that supervised machine learning approaches could be used to try and find \textsc{ML}-based approximations of the minimum fuel required for the hop in the low-thrust world. We constructed a database of about 800,000 transfers, which were found by solving the low-thrust fuel-optimal control problem and we trained a supervised neural network to predict the ratio between the final and initial mass (which is directly related to the propellant cost of the transfer, via the Tsiolkovsky equation). The objective was, given a rendezvous transfer between two asteroids, with fixed initial and final time (and therefore time of flight) and fixed initial mass, to find the minimum propellant required by the transfer. 

We used a feed-forward neural network with two hidden layers and 50 nodes per layer. As attributes we included the cost of the Lambert transfer between the two asteroids as well as the time of flight, the initial mass, the difference of their eccentricities, mean motions, and the semi-major axis. A main concern was to ensure that the network treated equally hops between asteroids that were at the same distance from the Sun and were defined by the same transfer geometry. To achieve this, instead of feeding as attributes to the network the position and velocity of the two asteroids in the inertial frame, we used the relative position and velocity and we additionally expressed their coordinates as spherical coordinates in the LHLV frame (Local Horizontal Local Vertical) attached to the departure asteroid.
%Our first goal was to ensure that the network treated equally hops between asteroids that were at the same distance from the Sun and the same transfer geometry, but rotated. To achieve this, instead of feeding as attributes to the network the position and velocity of the two asteroids, we passed their relative position and velocity, expressing their coordinates as spherical coordinates with respect to an orbital frame that is located at the first asteroid, with the x-axis aligned along the radial, the z-axis perpendicular to the orbital plane of the asteroid, and the y axis that completed the orthogonal reference frame. 
Furthermore, we also used as attributes both the impulsive velocity vectors derived from the difference of velocity between the Lambert arc solution and the asteroids' velocity (this corresponds to the velocity vectors needed to perform the Lambert transfer that enables the rendezvous). These represent a fast and informative guess of the solution and were found to greatly improve the results. Both were also expressed using the spherical coordinates with respect to the same reference frame as the relative position and velocity. In Figure~\ref{figure:ml_low_thrust}, we display a schematic representation of the neural network architecture and inputs/outputs used in the competition.

\begin{figure}[htp]
  \centering
  \begin{tikzpicture}[
      wire/.style = {
        ->, semithick, >={[round,sep]Stealth}
      },
      neuron/.style = {
        draw=peri_color,
        fill=peri_color!25,
        circle,
        thick,
        baseline
      },
      input/.style = {
        semithick,
        rectangle,
        fill = hop_color!50,
        rounded corners,
        inner sep = 0pt,
        outer sep = 0pt,
        minimum width=1.5cm,
        minimum height=0.6cm,
        baseline
      }
    ]

    % Input layer
    \node[input] (box1) at (0, 0)   {$\mathbf{r}_2-\mathbf{r}_1$};
    \node[input] (box2) at (0,-1.5) {$\mathbf{v}_2-\mathbf{v}_1$};
    \node[input] (box3) at (0,-3)   {$\Delta \mathbf{v}_{L,1}$};
    \node[input] (box4) at (0,-4.5) {$\Delta \mathbf{v}_{L,2}$};
    \node[input] (box5) at (0,-6)   {$m_0$};

    % Fully connected layers
    \node[neuron,              right = 2cm of box1] (fc1-1) {};
    \node[neuron,              right = 2cm of box2] (fc1-2) {};
    \node[neuron,              right = 2cm of box3] (fc1-3) {};
    \node[neuron, draw = none, fill = none, right = 1.75cm of box4] (fc1-5) {$\vdots$};
    \node[neuron,              right = 2cm of box5] (fc1-4) {};

    % Fully connected layers
    \node[neuron,              right = 4cm of box1] (fc2-1) {};
    \node[neuron,              right = 4cm of box2] (fc2-2) {};
    \node[neuron,              right = 4cm of box3] (fc2-3) {};
    \node[neuron, draw = none, fill = none, right = 3.75cm of box4] (fc2-5) {$\vdots$};
    \node[neuron,              right = 4cm of box5] (fc2-4) {};

    % Output layer
    \node[right=2cm of fc2-3] (output) {$\dfrac{m_f}{m_0}$};

    % Arrows
    \begin{scope}[on background layer]
      \foreach \i in {1,...,4}{
        \foreach \j in {1,...,4}{
          \draw[wire] (box\i)  -- (fc1-\j);
          \draw[wire] (fc1-\j) -- (fc2-\i);
          \draw[wire] (fc2-\i) -- (output);
        }
      }
      \foreach \i in {1,...,4}{
        \draw[wire] (box5)  -- (fc1-\i);  
      }
    \end{scope}
  \end{tikzpicture}    
  \caption{Schematic representation of neural network architecture used to learn the low-thrust transfer cost. Vectors are all expressed in spherical coordinates in the orbital frame of the departure body to ensure rotational symmetry.}
  \label{figure:ml_low_thrust}
\end{figure}
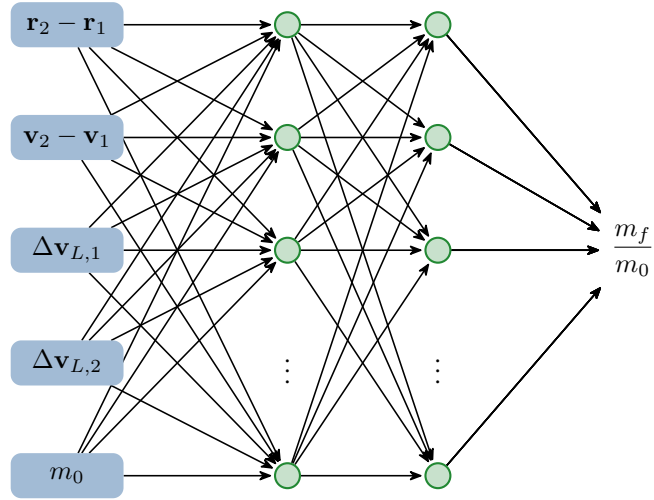

Finally, all the attributes, together with the initial mass and time of flight, were normalized. We trained the network for 150 epochs, using a batch size of 1024 and a learning rate of 0.001. 
In Figure~\ref{fig:ml_vs_analytical_dv_approx}, we show the results of the \textsc{ML} approximation compared to the analytical approximation across the 800,000 transfers in terms of absolute error (in kg) of the predicted mass, compared to the fuel-optimal control problem solution.

\begin{figure}
  \centering
  \begin{subfigure}[b]{0.47\textwidth}
    \includegraphics[width=\textwidth]{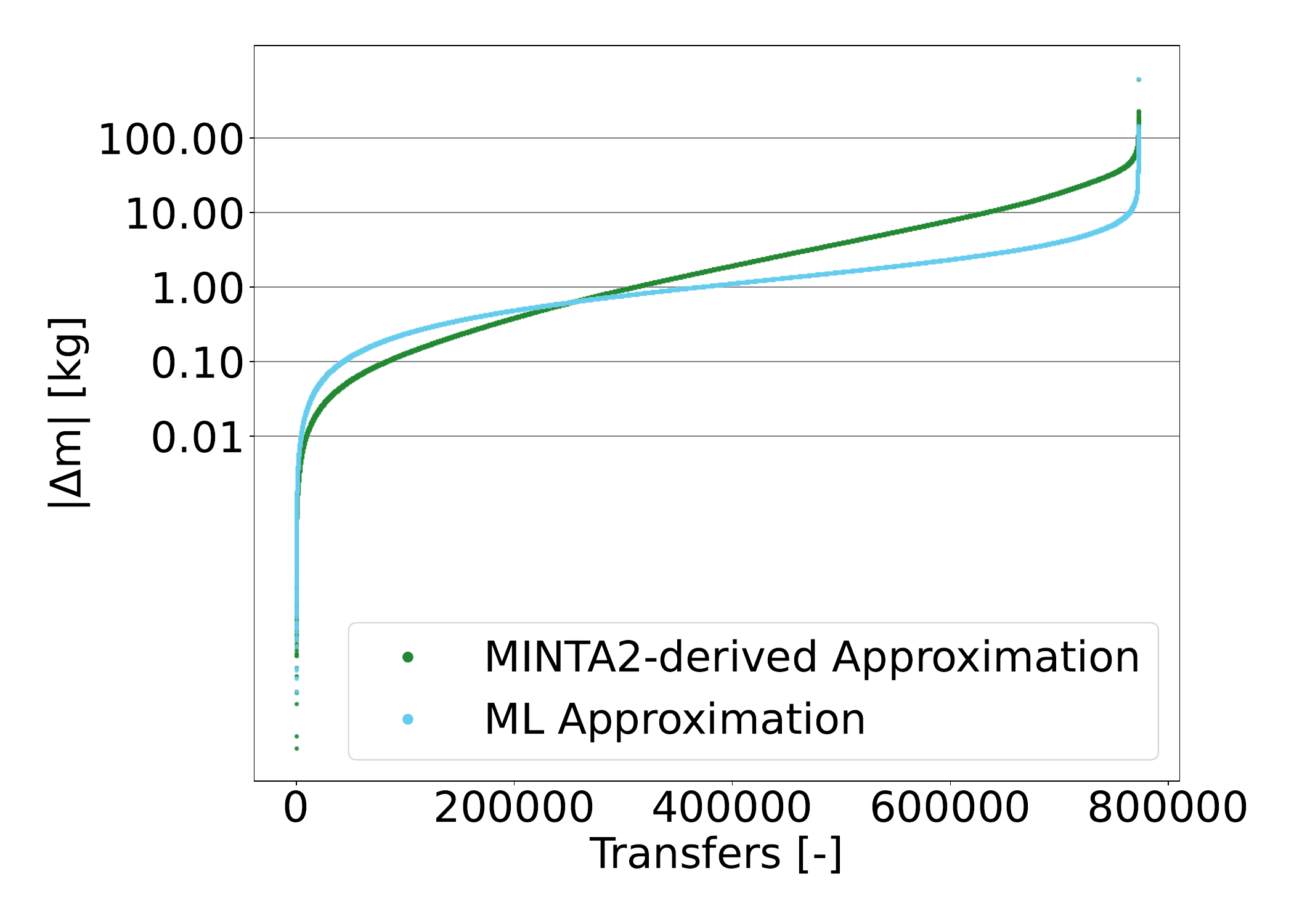}
    \caption{}\label{fig:ml_vs_analytical_dv_approx}
  \end{subfigure}
  \hfill
  \begin{subfigure}[b]{0.48\textwidth}
    \includegraphics[width=\textwidth]{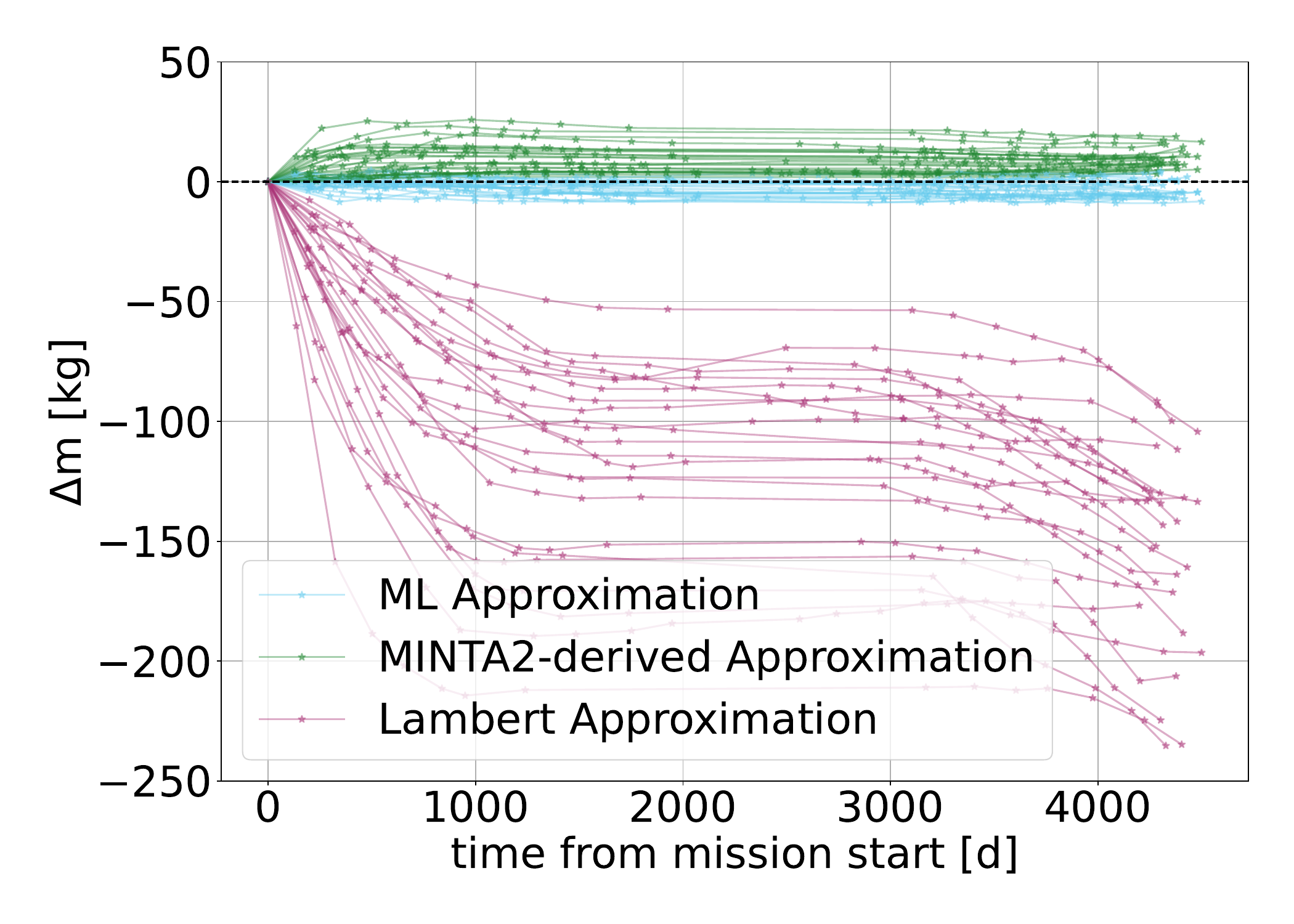}
    \caption{}\label{fig:ml_vs_minta2_vs_lambert}
  \end{subfigure}
  \caption{\textit{Left}: Number of low-thrust optimal transfers approximated within some error using the \textsc{MINTA2}-derived and the analytical approximations. \textit{Right:} Mass error accumulated along the time of flight by 20 different self-sufficient ships when the mass is approximated using Lambert, \textsc{ML}-based approximation and \textsc{MINTA2}-derived approximation.}
\end{figure}

As we can see, for errors below around 800 grams in the final mass, the analytical approximation has the upper hand on the \textsc{ML} approximations. However, and most importantly, above 800 grams of absolute errors, the neural network manages to significantly reduce the errors (e.g. for more than 95\% of the transfers the errors are kept below 10 kg on the final mass). Furthermore, the analytical approximation can have large errors up to 100 kg and would have dire consequences in any search strategy, causing substantial miscalculations on the propellant cost and the overall chain quality.
Furthermore, the inference speed to predict from a hop and an initial mass the final mass is estimated to be about $ \times1.8$ faster for the \textsc{ML} approximator than for the analytical one. Due to the quality and speed of the results, we decided to employ the \textsc{ML}-based approximator to estimate the final mass (and hence the propellant cost) during the tree searches. 
For a comprehensive evaluation of the performance of various implemented approximations, we present in Figure~\ref{fig:ml_vs_minta2_vs_lambert} the mass errors incurred during a sequence of asteroid hops across diverse self-sustaining spacecraft.
Specifically, we analyze 20 ships that made it to our final solution, which will be discussed in detail in Section~\ref{sec:final_assembled_solution}. Each spacecraft undertakes between 16 and 20 hops: we show the error made in the mass computations along the full trajectory using distinct approximators: Lambert approximation, \textsc{MINTA2}-derived approximation, and \textsc{ML} approximation. Notably, the \textsc{ML} approximation applied to approximately 20 spacecraft, exhibits cumulative error values below 10-15 kg even after more than 15 hops. Similarly, the \textsc{MINTA2}-derived one shows marginally higher error values, reaching up to approximately 25 kg.  In contrast, the Lambert approximation yields errors as substantial as 240 kg. Interestingly, the errors in Lambert approximations tend to escalate more prominently in the initial hops, while displaying a more stabilized behavior in subsequent hops. It is essential to note that this observation does not establish a general pattern for all asteroid hop sequences. However, it suggests that the combinatorial algorithm generating these solutions prioritizes time optimality in the initial hops and leans towards mass optimality in the later. In fact, the Lambert solution is expected to exhibit poor performance, especially closer to time optimality.

%\begin{tikzpicture}
%    % Define two points in space in red and add a single dot at the center
%    \coordinate (point1) at (0,0);
%    \coordinate (point2) at (4,2);
%    \fill[red] (point1) circle (2pt);
%    \fill[red] (point2) circle (2pt);
    
%    % Draw orthogonal reference frame at the first point
%    \draw[->] (point1) -- ++(1,0) node[right] {$\hat{t}$};
%    \draw[->] (point1) -- ++(0,1) node[above] {$\hat{h}$};
%    \draw[->] (point1) -- ++(45:0.6) node[above right] {$\hat{r}$};
    
    % Add vectors
%    \draw[->,blue] (point1) -- ++(1.5,0) node[midway, below]{$\Delta \vec{v}_{1,L}$};
%    \draw[->,blue] (point2) -- ++(1.2,0.8) node[midway, below]{$\Delta \vec{v}_{2,L}$};
    
    % Draw a thicker dotted open conic trajectory path in black, passing through the red dot and shifted slightly north
%    \draw[black, loosely dotted, thick, rotate=-90] (point1) to[out=30, in=195] (2.5,1.2);
    
    % Add a vector connecting the first red point to the second point and move it 1 inch upwards
%    \draw[->,gray, ultra thick, shift={(0,1in)}] (point1) -- (point2) node[midway, above]{$\vec{r}_2 - \vec{r}_1$};
%\end{tikzpicture}

\section{Design of Self-sufficient Ships}
\label{sec:oneship}
An ensemble is here defined as a collection of ships that visit asteroids in such a way that no asteroid is visited more than twice. One ship ensembles, here also called self-sufficient ships, are sub-optimal for this problem but offer some distinct advantages in general, and is what we ended up designing. 
The search for suitable trajectories for isolated ships can be massively parallelised without the need for synchronisation. 
In addition, it becomes simpler to manage the overall ensemble since any asteroid visited \textit{once} (for miner deployment) can be simply put in a set of asteroids that are out of limits to all other ships. 
The flip side is that if the search is executed in parallel, it becomes necessary to devise a method to compile the final solution in order to ensure that the chosen ships have not visited the same asteroid by chance. We have developed and used two different search strategies to design self-sufficient ships, two approaches we call \lq\lq Trajectory scaffolding\rq\rq\ and a \lq\lq A Beam Search that Looks Ahead in Time\rq\rq.

\subsection{Trajectory Scaffolding}\label{sec:scaffolding}
This approach is based on the observation that a ship deploying a miner on an asteroid can remain in the same orbit as the asteroid without expending any fuel -- referred to as `coasting' -- until some suitable point in the future when it can collect the material and return to Earth. Keeping in mind the fact that performing these manoeuvres (deployment, coasting, and retrieval) only \textit{once} is wasteful, we concentrate on finding a sequence of efficient trajectories (Asteroid1 -- Asteroid2 -- Asteroid1, i.e., $\alpha_1, \alpha_2, \alpha_1$) that maximise the returned material; this corresponds to maximising the coasting time as mining time is directly proportional to the amount of mined material. We refer to this type of trajectory as a scaffold.

Let us denote an initial scaffolding manoeuvre to the $i^{th}$ asteroid (Earth -- $\alpha_{i}$ -- Earth) as $E\rightarrow\alpha_{i}\rightarrow{E}$. A ship would be able to perform the manoeuvre and return to Earth with most of its propellant unused. Combining this with the long coasting time provides the ship with the opportunity to perform another scaffolding manoeuvre from $\alpha_{i}$ to another asteroid ($\alpha_{j}$). Provided that the ship can return from $\alpha_{j}$ back to $\alpha_{i}$ \textit{before} the departure time for the leg $\alpha_{i}\rightarrow{E}$, the ship would be able to collect as much material as has been mined during the coasting time at $\alpha_{i}$ as well as the material mined from $\alpha_{j}$. Then, from $\alpha_{j}$ the ship can perform another detour to another asteroid $\alpha_{k}$ and so forth, building a series of round trips -- a \textit{scaffold} -- that has the original $E\rightarrow\alpha_{i}\rightarrow{E}$ round trip as a common base (Fig. \ref{fig:scaffolding}).

\begin{figure}[htp]
    \centering
    \includegraphics[width=0.85\linewidth]{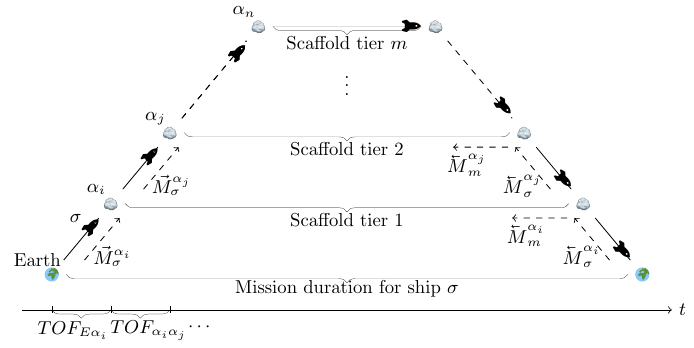}
    \caption{Illustration of the trajectory scaffolding approach. Here, $\alpha_{(\cdot)}$ indicates an asteroid, $\sigma$ indicates a ship, and $\arrow{M}_{\sigma}$ and $\worra{M}_{\sigma}$ denote the total mass of ship $\sigma$ for the forward and backward legs, respectively. The mass estimate for the backward leg is required in order to estimate the time of flight (ToF) for each forward and backward leg with the \textsc{MINT} approximator. Refer to the text for details.}
    \label{fig:scaffolding}
\end{figure}
\subsubsection{Building a Scaffold}

As with all search strategies considered here, when building a scaffold, a ship should arrive from Earth as early as possible and leave back for Earth as late as possible in order to maximise the mining time. In addition, the symmetric nature of the round trip imposes the additional requirement that for each tier $\alpha_{i}\rightarrow\alpha_{j}\rightarrow\alpha_{i}$ in the scaffold, a feasible return hop can be made from $\alpha_{j}$ to $\alpha_{i}$ \textit{before} the scheduled departure time for the hop $\alpha_{i}\rightarrow{E}$. The set of detour candidates $\mathcal{A}_{j}$ are those asteroids that are reasonably close to the current asteroid $\alpha_i$ at both arrival and departure time. We use the orbital indicator \cite{hennes2016fast, Izzo2016Designing} to find suitable asteroids for detours.

The first tier of the scaffold is built by selecting an initial asteroid $\alpha_{i}$ and choosing the forward hop from Earth to the asteroid (${E}\rightarrow\alpha_{i}$) and the hop back to Earth ($\alpha_{i}\rightarrow{E}$). These hops are chosen from among the options in the databases of arrival and return hops developed during the competition. % TODO: add reference
The time of flight (ToF) for each hop (forward or backward) is assumed to be the minimum time of flight (MINT) as approximated by our MINTA2 approximator (see \S \ref{sec:analytical_approximations}), with an added slack of $10$ days to avoid solutions that are too aggressive in time optimality. Ultimately, the goal is to maximise the coasting window at the next detour (which is equivalent to minimising the TOF) while also minimising $\Delta V$. However, there is one additional detail: while we know the mass of the ship for the forward hop, we must estimate the starting mass for the backward hop as this is required to estimate the MINT for the return hop.

\begin{lstlisting}[
    style=pseudocode,
    backgroundcolor=,
    label=algo:scaffold,
    float=htpb,
    caption={Algorithm for building a trajectory scaffold $\mathcal{S}$ for a single ship $\sigma$. Here, $C$ denotes a chain of transfers that constitute a scaffold. Here, $\alpha_{i}$ is the first asteroid $M^{\alpha_{i}}_{DB}$ denotes the mass of the ship at the moment of arrival at asteroid $\alpha_{i}$, and $t^{\alpha_{i}}_{in}$ and $t^{\alpha_{i}}_{out}$ denote the times of arrival at and departure from the current asteroid $\alpha_{i}$. The current forward mass and the estimated backward mass of the ship $\sigma$ are denoted as $\arrow{M}_{\sigma}$ and $\worra{M}_{\sigma}$, respectively. $M_{d}$ denotes the dry mass of the ship, and $M^{\alpha_{(\cdot)}}_{m}$ denotes the material collected from asteroid $\alpha_{(\cdot)}$. $S$ denotes the slack (in days) added to the $ToF$ for each hop.}
  ]
  algorithm $\textsf{scaffold}(\alpha_{i},M^{\alpha_{i}}_{DB},t^{\alpha_{i}}_{in},t^{\alpha_{i}}_{out},ToF_{min}=30,ToF_{max}=250,S=10)$
    $C$ <- [] // Chain of scaffolding tiers (excluding hops from and to Earth, which are bootstrapped)
    $\alpha$, $t_{in}$, $t_{out}$, $\arrow{M}_{\sigma}$ <- $\alpha_{i}$,$t^{\alpha_{i}}_{in}$,$t^{\alpha_{i}}_{out}$, $M^{\alpha}_{DB}$
    while $t_{in} < t_{out}$ do
        FW <- Top 500 closest asteroids to $\alpha$ at time $t_{in}$ // Forward transfer candidates
        BK <- Top 500 closest asteroids to $\alpha$ at time $t_{out}$ // Backward transfer candidates
        $\mathcal{A}_{j}$ = FW $\bigcap$ BK // Keep asteroids that are in both sets
        for each $\alpha_{j} \in \mathcal{A}_{j}$
            for $ToF \in (ToF_{min}, ToF_{max})$ // Partition the ToF
                if (forward) // Mass for the forward hop
                    $M_{\sigma}$ <- $\arrow{M}_{\sigma}$
                else // Mass for the backward hop
                    $M_{\sigma}$ <- $\frac{\arrow{M}_{\sigma} + M_{d}}{2}$ + $M^{\alpha_{i}}_{m}$
                $\rho^{\alpha}_M$ <- Mfm0($M_{\sigma}$, $\alpha$, $t_{in}$, $t_{out}$)
                if ($M_{\sigma}$ < MIMA($\alpha$, $t_{in}$, $t_{in} + ToF$) and 0.5 < $\rho^{\alpha}_M$ < 1
                    $\arrow{M}_{\sigma}$ <- $\rho^{\alpha}_M$$\arrow{M}_{\sigma}$ - 40 kg // Depositing a miner reduces the mass by 40 kg
                    $\alpha$ <- $\alpha_{j}$
                    if (forward) // Add the TOF to the time of arrival at $\alpha_{i}$
                        $t_{in}$ <- $t_{in}$ + TOF + S
                    else // Subtract the TOF from the time of departure from $\alpha_{i}$
                        $t_{out}$ <- $t_{out}$ - TOF - S
                    Append hops [$\alpha_{i}\rightarrow\alpha_{j}$ and $\alpha_{j}\rightarrow\alpha_{i}$] to $C$
                    break
    return $C$
\end{lstlisting}

\begin{figure}[htp]
    \centering
    \includegraphics[width = \linewidth]{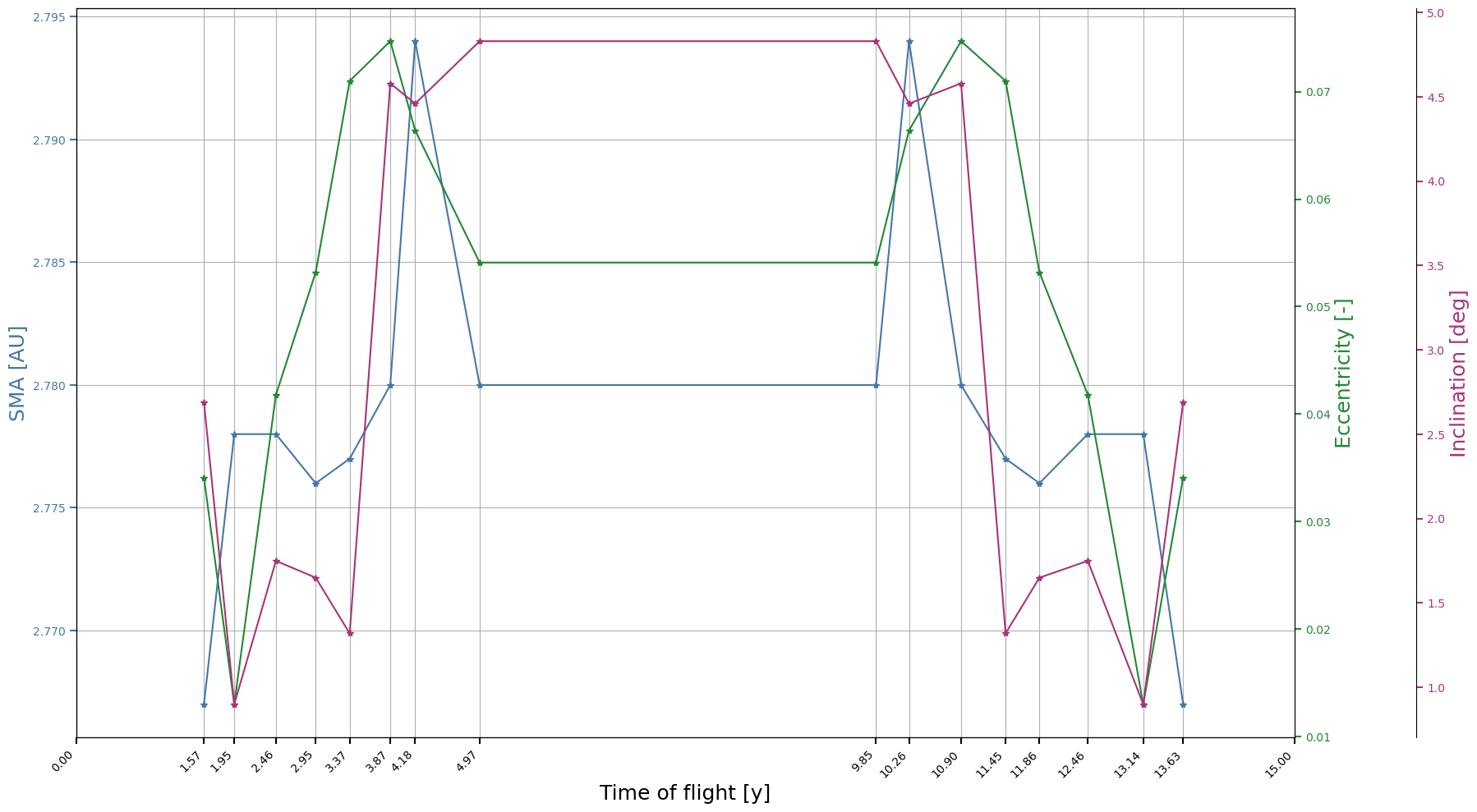}
    \caption{Orbital parameters (semi-major axis, eccentricity, and inclination) and for the orbits of the asteroids visited during the mission window for the chain at the bottom of Fig. \ref{fig:r_of_t_plots}. Note the symmetrical trajectory pattern that emerges naturally from the fact that during the collection phase (after coasting) the asteroids are visited in the opposite order to that of the miner deposition phase.}
    \label{fig:scaffold_example}
\end{figure}

The na\"ive way to estimate the mass of the ship for the backward hop is to subtract the fuel consumed during the forward hop, the mass of a miner, and an amount of material proportional to the coasting time. This can be done because once the ship has arrived at the new asteroid, it is assumed to be coasting without expending any fuel. However, this estimate is in fact quite inaccurate since presumably the ship would be performing multiple detours, thus expending more fuel and accumulating more material along the way, and therefore the actual mass at the end of all detours would look very different from this na\"ive approach.

Therefore, it is necessary to estimate the mass of the ship at each backward hop (Fig. \ref{fig:scaffolding}). Assuming that the ship is currently at asteroid $\alpha_{i}$, the backward mass $\worra{M}^{\alpha_{i}}_{\sigma}$ is estimated as follows based on the current amounts of fuel and material on board:

\begin{align}
    \worra{M}^{\alpha_{i}}_{\sigma} &= \arrow{M}^{\alpha_{i}}_{\sigma} - \frac{M_{f}}{2} + M^{\alpha_{i}}_{m}
\end{align}

Here, $M^{\alpha_{i}}_{\sigma}$ indicates the total current mass of the ship, $M_{f} = \arrow{M}^{\alpha_{i}}_{\sigma} - M_{d}$ is the mass of the fuel and any miners onboard ($M_{d}$ is the dry mass of the ship, fixed at $500~kg$), and $M^{\alpha_{i}}_{m}$ and $M^{\alpha_{j}}_{m}$ are the amounts of material collected at asteroids $\alpha_{i}$ and $\alpha_{j}$, respectively. This estimate assumes that the ship can perform all future detours and return to the current asteroid with half of its remaining fuel, and allows us to determine the feasibility and the ToF for a backward transfer. Note that since we do not know in advance how many detours the ship would be able to make, $M_{f}$ also includes the mass of an undetermined number of miners onboard.

For each new detour, ToFs in the range $(30, 250)~d$ are tested for feasibility\footnote{Naturally, the TOF for the backward hops is \textit{subtracted} from the departure point for the current asteroid.} by checking if the current mass is below the approximated MIM (Eq. \ref{eq:mima}) for the respective ToF. A fixed buffer of $10$ days is added to avoid transfers that were too aggressive and to absorb errors resulting from the ML approximators and the fact that the mass of miners is included in the total estimate of the fuel mass. If the trajectory is feasible, the ratio of the ship's mass after and before the hop is obtained from the ML-based approximator:

\begin{align}
    \arrow{M}^{\alpha_{j}}_{\sigma} &:= \rho^{ij}_{M}\arrow{M}^{\alpha_{i}}_{\sigma}
\end{align}

Here, $\rho^{ij}_{M}$ denotes the ratio of the final to initial mass ($\frac{m_f}{m_{0}}$; cf. Fig. \ref{figure:ml_low_thrust}) for the transfer from asteroid $\alpha_{i}$ to asteroid $\alpha_{j}$. An additional check ensures that $\rho_{M}\in(0.5,1)$ to avoid exceedingly inefficient transfers (cases where $\rho_{M}\le{0.5}$) and erroneous outputs from the ML approximator (cases where $\rho_{M}\ge{1}$). If the checks pass, $\rho_{M}$ is used to compute the mass after the forward hop. This procedure is then repeated with the new asteroid $\alpha_{j}$ as the top tier of the scaffold until no feasible detour can be found. A detour is deemed infeasible when the departure time from the new detour candidate is evaluated to be \textit{before} the arrival time at that candidate, a contradiction. The algorithm for building a trajectory scaffold in this way is given in Listing \ref{algo:scaffold}. 

The resulting trajectory scaffolds are further optimised due to the cumulative estimation error for the backward mass. % TODO: Is the optimisation metod in here?
Despite this, scaffolds produced diverse trajectories that proved useful when constructing larger solutions. An example of one such scaffold is illustrated in Fig. \ref{fig:scaffold_example}.

\subsection{A Beam Search that Looks Ahead in Time}\label{sec:phoenix}

An alternative approach, distinct from the one discussed earlier, relies on a modification to the standard Beam Search approach as outlined in \cite{Izzo2016Designing}. In the following, we formally define the problem within the context of graph theory. Subsequently, we will delve into the algorithmic details of our heuristic solution.

A \emph{configuration} of a spacecraft is an element of
$\mathcal{A}\times\mathbb R\times\mathbb R$, that is, an
asteroid at a given time with a given amount of (remaining)
mass. Define the \emph{configuration graph} of the mission as the
directed graph~$D$ with vertex set
\(
V(D) = \big(\mathcal{A}\times\mathbb R\times\mathbb R\big)\mathrel{\dot\cup}\{E_{s}, E_t\}
\)
with $E_s$ and $E_t$ being two distinct vertices that represent the
start and end at Earth. The (directed) edge set of $D$ is defined as:
\begin{align*}
  E(D) &= \{\, (x,y) \mid \text{$x,y\in\mathcal A\times\mathbb R\times\mathbb R$ and there is a \emph{hop} from $x$ to $y$} \,\}\\
  &\qquad\cup\,\{\, (E_s,x) \mid x\in\mathcal{D}_{\text{dep}}\,\}
  \,\cup\,\{\, (x,E_t) \mid x\in\mathcal{D}_{\text{ret}}\,\}.
\end{align*}
Here, a \emph{hop} from $(\alpha,t,m)$ to $(\alpha',t',m')$ is a
trajectory that brings a spacecraft that is rendezvoused with $\alpha$
at time $t$ with a mass $m$ to rendezvous with $\alpha'$ at time $t'$
with a remaining mass of $m'$. Note that the graph $D$ is
infinite. However we may assume that $D$ is given via an \emph{oracle} that tells us whether a given hop is possible, which formally is a function:
\[
\mathrm{hop}\colon\mathcal A\times\mathbb R\times\mathbb R\times\mathcal A\times\mathbb R\times\mathbb R\rightarrow\{0,1\}.
\]
The concept of a \lq\lq magic\rq\rq\ oracle is well known and used in graph theory to formalize the mathematical structure of graph problems without focusing much on the specific problem domain. In our case the oracle encapsulate all space-flight mechanics knowledge that must be used in order to conclude whether a certain randezvous (hop) is actually feasible and can be performed (in which case the edge exists) or not (in which case the edge is not there). The results from Section~\ref{sec:lt2impulse}, as well as some book-keeping of the various masses collected, allow to code efficiently such an oracle in practice, without solving any optimal control problem.
Observe that $D$ is
acyclic, as all edges go ``along with time.''  A path from $E_s$ to
$E_t$ in $D$ corresponds to a \emph{multi-rendezvous} trajectory of a
spacecraft that launches from Earth and returns to Earth at the
end of the mission.

Of course, not all $E_s$-$E_t$-paths are equally
contributing towards the mission's objective. A common desire in multi-rendezvous missions is to
rendezvous with as many asteroids as possible, which means that we
need to find an $E_s$-$E_t$-path of maximum
length such that no two vertices share the same $\mathcal
A$-coordinate. This is a well-known $\Class{NP}$-complete problem, even if
$D$ is finite and entirely given in the input.\footnote{The attentive
reader may notice that $D$ is acyclic and that, in acyclic graphs,
the longest path problem can actually be solved in polynomial
time. This is in general true, however, the additional constraint that
we are not allowed to visit two vertices with the same $\mathcal
A$-coordinate renders the problem $\Class{NP}$-hard again.} The now \emph{de
facto} standard approach for solving multi-rendezvous problems is, thus, a
heuristic version of a best-first search, called \emph{beam
search}~\cite{simoes2017multi}. The beam search explores $D$ starting from
$E_s$ in the following sense: it maintains a list of active
vertices (called the \emph{beam}), which initially contains
$E_s$. Then, the search is performed in phases, whereby
every phase consists of two steps that explore $D$ layer-by-layer:
\begin{description}
  \item[Expansion] first, every element in the beam is replaced by its
    neighborhood in $\mathcal A\times\mathbb R\times\mathbb R$;
  \item[Reduction] then, a \emph{score} is computed for every element
    in the beam and only the top $\beta$-elements are kept in the beam
    (hence, the beam has size at most $\beta$ at the end of every phase).
\end{description}
This process is repeated until no element in the beam has any neighbor
(because the configurations ran out of time or propellant). 
During the search, we keep track of configurations from which there is an edge to
$E_t$ so that, at the end, we can output the best one found. 
The parameter $\beta\in\mathbb{N}$ is called the \emph{beamwidth} and
(roughly) describes the quality of the search. If we set $\beta=1$,
the algorithm becomes a simple greedy search, and with increasing
$\beta$, it converges towards an exact brute-force algorithm. 
Two technical details omitted in the above description are that, first, two configurations could lead to the same successor configuration and, thus, that we may explore the same part of $D$ multiple times. 
This can be avoided by cacheing already explored configurations. 
Second, a vertex in $D$ has an exceedingly large number of neighbors, which we can only partially explore. Hence, some heuristics must be used to reduce the neighborhood during the expansion step.

The parameter $\beta$ (maximum beam size) can usually be determined relatively easily through
experiments, while the heuristic to select neighbors falls into the domain of mission analysis and is heavily problem-dependent. 
In our experiments, for each element in the beam 1,000 target asteroids are selected using the orbital phasing indicator described in~\cite{simoes2017multi}. 
From these, a set of potential neighbours (low-thrust hops) is generated by gridding time from the minimum transfer time MINT (as computed by the MINTA2 approximation) for the following year in months increments. \footnote{The grid [50, 100, 150, 200, 69$\pi$, 250] days was used in the last version of our experiments for reasons that are no longer obvious.}

From an algorithm design perspective, a most interesting parameter of the algorithm is the \emph{score}
function used in the reduction step. If the goal is to ``just'' find a
path that is as long as possible, reasonable scores are the amount of
remaining time or mass, or a weighted combination of the two (see \cite{Izzo2016Designing} for an extended discussion on multi-objective beam search in the context of mission analysis). However,
this reasoning fails for the mission objective of GTOC12:
Since a miner needs to be deployed \emph{and} collected, the ships
need to visit the same asteroid \emph{twice}, at different epochs. This is precisely a
property for which the beam search does \emph{not} work well~--~the
heuristic is supposed to make a locally good decision after all. We
circumnavigated this issue by introducing the score shown in Listing~\ref{algo:yetanotherscore}, which tries to capture how easily the spacecraft can return, at some point in the future, to the same asteroid. 

Formally, let $f$ be a function that
rates the quality of a hop~--~for instance, $f(h)$ could be the 
$\Delta v$ required by the optimal low-thrust transfer for hop $h$ or an approximation (the Lambert or the machine learned one); or it
could be a phasing indicator~\cite{Izzo2016Designing, hennes2016fast} i.e., not accounting for $t'$ at all. 
Let us denote for a hop $h$ the \emph{reverse hop} with $\mathrm{rev}(h)$, i.e., if $h$ is a hop from $\alpha$ at time $t$ to $\alpha'$ at time $t'$, then $\mathrm{rev}(h)$ is the hop from $\alpha'$ at time $t$ to $\alpha$ at time $t'$ (masses are also swapped but are not relevant in what follows). 
Finally let $h\oplus T$ denote, for a hop $h$ and a time duration $T$, the same hop in the future (e.g., $h\oplus 365.25$ would be the same hop in one year).

\begin{lstlisting}[
  style=pseudocode,
  backgroundcolor=,
  label=algo:yetanotherscore,
  float=htb,
  caption={A scoring algorithm for a configuration
    $(\alpha_i,t_i,m_i)$ based on a function $f$ that measures the
    quality of a hop (for instance, the $\Delta v$ of an impulsive
    Lambert transfer). The value $q_1$ is the quality of the hop used
    to get to the configuration, and $q_2$ is the best quality a
    reverse hop has in the feature (at some fixed sampled dates).}
  ]
algorithm $\textsf{future-score}(\alpha_i,t_i,m_i)$
  $h$ <- the hop that lead the beam search to $(\alpha_i,t_i,m_i)$
  $q_1$ <- $f(h)$
  $q_2$ <- $min\{ f( \mathrm{rev}(h)\oplus T\cdot365.25) \mid T\in\{3,\dots,9\} \}$
  return $q_1+q_2$
\end{lstlisting}

Figure~\ref{figure:beamsearch-hist} shows, as an example, the distribution of
collected masses of ships found using the beam search with the score
described in Listing~\ref{algo:yetanotherscore} for the case that $f(h)$
is computed as an impulsive Lambert
transfer (left) and using the orbital indicator~\cite{Izzo2016Designing,
hennes2016fast} (right). Overall, to create a large pool of promising self-sufficient ships, we performed multiple runs of the beam search: the histograms reported in Figure~\ref{figure:beamsearch-hist} are only an example of two runs that were performed at some point during the competition. Our experiments indicated that the scoring function $f$ based on the orbital indicators performed the best on average. This was interpreted as an indication that the reverse hop score benefits from not depending explicitly on $t'$.

\begin{figure}[htb]
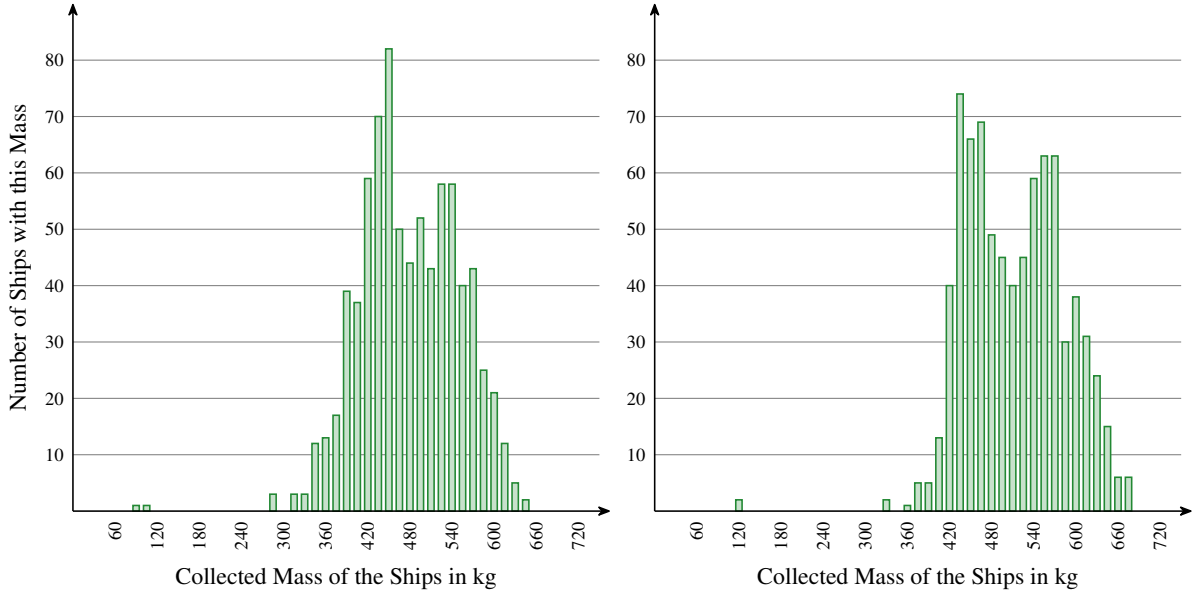

  \centering
    \mbhist{0/15,0/30,0/45,0/60,0/75,1/90,1/105,0/120,0/135,0/150,0/165,0/180,0/195,0/210,0/225,0/240,0/255,0/270,3/285,0/300,3/315,3/330,12/345,13/360,17/375,39/390,37/405,59/420,70/435,82/450,50/465,44/480,52/495,43/510,58/525,58/540,40/555,43/570,25/585,21/600,12/615,5/630,2/645,0/660,0/675,0/690,0/705,0/720}{x}
  \mbhist{0/15,0/30,0/45,0/60,0/75,0/90,0/105,2/120,0/135,0/150,0/165,0/180,0/195,0/210,0/225,0/240,0/255,0/270,0/285,0/300,0/315,2/330,0/345,1/360,5/375,5/390,13/405,40/420,74/435,66/450,69/465,49/480,45/495,40/510,45/525,59/540,63/555,63/570,30/585,38/600,31/615,24/630,15/645,6/660,6/675,0/690,0/705,0/720}{}
  \caption{Distribution of the masses collected by ships found via one run of the beam search with the \textsf{future-score}. In the left plot, we show the distribution for the case that $f$ simply computed the required
  $\Delta v$ of the hop as impulsive Lambert transfer; in the right
  plot, $f$ measures the quality of a hop using the orbital indicator~\cite{Izzo2016Designing,
  hennes2016fast}.}
  \label{figure:beamsearch-hist}
\end{figure}

\section{Selecting the Best Miners Subset}
\label{sec:patch}
% In Section~\ref{sec:diversify} we briefly sketch our efforts to counter both of the above challenges. Section~\ref{sec:combining} describes our algorithm to combine ships to a full solution such that the score function of the competition is maximized. Finally, Section~\ref{sec:splicing} outlines an approach to disentangle dependent ships in certain cases.

%\subsection{Diversifying the fleet and dodging the bonus function}\label{sec:diversify}

Since the algorithm developed in Section~\ref{sec:phoenix} is a heuristic,
we can generate a diverse pool of single ships by varying the initial conditions. All in all, we found 371 \emph{promising} ships, i.\,e., ships that collected a mass of at
least 580\si{kg} and as high as 704\si{kg}. However, not all of these ships were independent, meaning that nothing prevented them from mining the same asteroid. Furthermore, asteroids that were already mined by previously submitted solutions suffered a penalty on their score according to the bonus function (see \cite{gtoc12}). Towards the final days of the competition, both the above factors increasingly countered our efforts to improve the score of our solution. We thus experimented with methods to have our search algorithms (beam search and scaffolding search) produce ships that were guaranteed to avoid these issues. These included the black-listing of single asteroids from the searched data set, but also the forced exploration of different regions of the search space than those to which our algorithms were by design attracted to, via targeted searches on small subsets of $\mathcal A$ containing around 200 asteroids. For instance, we observed a tendency towards approximately circular low-inclination trajectories for our best mass returning ships and hence searched subsets of $\mathcal A$ formed by asteroids with orbits in thin elliptical or inclined tori. As the end of the competition was approaching fast, these methods were ad-hoc and delivered mixed results, but ultimately contributed a few ships to our final solution.

\subsection{Combining Ships to a Full Solution of Maximized Score}
\label{sec:combining}

As explained, the pool of \emph{promising} ships, which exhibited dynamic growth and changes particularly in the concluding days of the competition, encompassed vessels with inherent incompatibilities and significant overlaps with selections made by other teams. 
This complexity rendered the straightforward selection of ships based on a naive collected mass ranking infeasible. 
To address this, we devised a method for identifying the optimal subset of ships from this pool. 
The objective was to maximize the cumulative collected mass, accounting for penalties, while adhering to the crucial non-overlap condition. This intricate optimization problem was effectively mapped to an Integer Linear Program (ILP) for systematic resolution.

Formally, we associate with a ship $\sigma_i$ a tuple $(m_i,s_i,\mathcal A_i)$ with
$m_i,s_i\in\mathbb{R}^+$ being the \emph{collected mass} of the ship and the
\emph{received score} (i.\,e., the mass after applying the bonus function),
respectively. Let further $\mathcal A_i\subseteq\mathbb{N}$ be the set of IDs of asteroids
visited by the ship.

Given a set $\mathcal{S}=\{\,(m_1,s_1,\mathcal A_1),\dots,(m_n,s_n,\mathcal A_n)\,\}$ of ships, our task is to
find a set $I\subseteq\{1,\dots,n\}$ such that:
\begin{enumerate}
\item $\sum_{i\in I}s_i$ is maximized;
\item $\mathcal A_i\cap \mathcal A_j=\emptyset$ for all $i,j\in I$
    with $i\neq j$;
  \item $\frac{1}{|I|}\sum_{i\in I}m_i\geq \frac{\log(|I|/2)}{0.004}$.
\end{enumerate}

That is, $I$ is a selection of ships (which we call an \emph{ensemble})
that maximizes the obtained score, and that is compatible in the sense
that no two ships share an asteroid. The third item is the
\emph{average constraint} imposed by the competition, i.\,e., to use a
certain number of ships, the average collected mass needs a certain
value. Note that we optimize the score while the average constraint
concerns the actual mass. See the left column in
Figure~\ref{figure:ensemble} for an example.

\begin{figure}[htp]
  \centering

  \begin{minipage}[t]{0.3\textwidth}
    \begin{align*}
      \mathcal{S} =\, &\{\\
      &\quad \sigma_1 = (300,1,\{a,b,c\}),\\
      &\quad \sigma_2 = (75,75,\{a,d\}),\\
      &\quad \sigma_3 = (75,75,\{b,e\}),\\
      &\quad \sigma_4 = (100,100,\{c,f\})\\
      &\}
    \end{align*}
  \end{minipage}
  \quad
  \begin{minipage}[t]{0.3\textwidth}
    \begin{tikzpicture}[baseline={(0,.5)}]
      \draw[color=gray, semithick] (-1.5, 0) -- (-1.5,-6);
      \draw[color=gray, semithick] ( 2.5, 0) -- ( 2.5,-6);

      \begin{scope}[yshift=-2cm, scale=1.25]
        \node (s1) at (0,0)  {$\sigma_1$};
        \node (s2) at (1,0)  {$\sigma_2$};
        \node (s3) at (0,-1) {$\sigma_3$};
        \node (s4) at (1,-1) {$\sigma_4$};
        \graph[use existing nodes, edges={semithick}]{
          s1 -- {s2,s3,s4};
        };
      \end{scope}
    \end{tikzpicture}
  \end{minipage}
  \quad
  \begin{minipage}[t]{0.3\textwidth}
    \begin{align*}
      &\text{Maximize}\\
      &1x_1 + 75x_2 + 75x_3 + 100 x_4\\
      &\text{Subject To}\\
      &x_1+x_2\leq 1\\[-1ex]
      &x_1+x_3\leq 1\\[-1ex]
      &x_1+x_4\leq 1\\[-1ex]
      &x_1,x_2,x_3,x_4\in\{0,1\}
    \end{align*}
  \end{minipage}
  \caption{This figure illustrates an example instance of the problem of constructing an ensemble. The left column shows a set of four ships; the middle column shows the corresponding conflict graph (containing an edge between two ships if they use a common asteroid), and the right column shows our base encoding.}
  \label{figure:ensemble}
\end{figure}
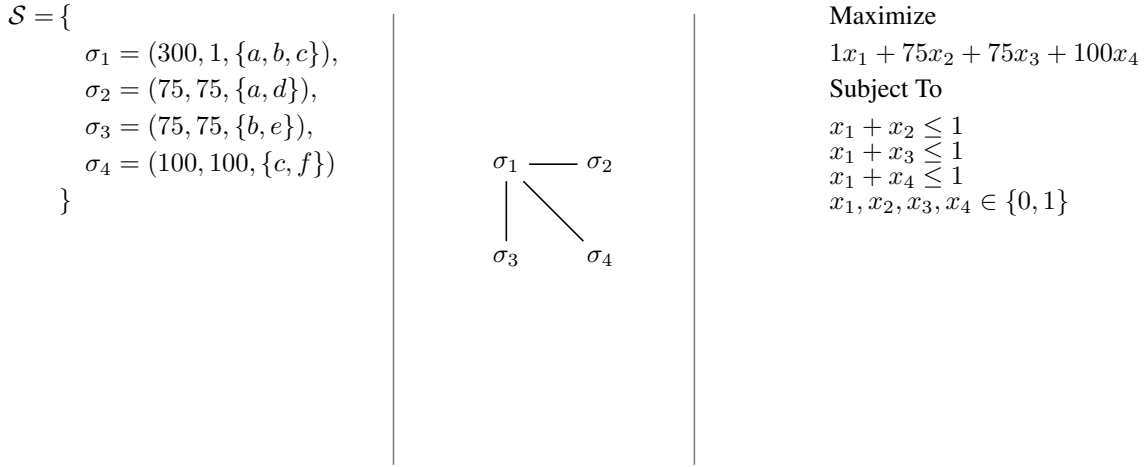

Our primary strategy to solve the problem is by encoding it as an
\emph{independent set} problem in the \emph{conflict graph}. In
detail, we construct a graph $G\coloneq G(\mathcal{S})$ with vertex
set $V(G)=\mathcal{S}$ and edge set
\[
  E(G)=\{\,\{(m,s,\mathcal A), (m',s',\mathcal A')\}\mid \mathcal A\cap \mathcal A'\neq\emptyset\,\}.
\]
The center column of Figure~\ref{figure:ensemble} shows an incarnation
of a conflict graph. Clearly, a score-optimal independent set (i.\,e.,
a set of pairwise non-adjacent vertices that maximizes the
corresponding score) is exactly a solution to the first two items of
the ensemble construction problem. In our running example, the
score-optimal independent set is $\{\sigma_2,\sigma_3,\sigma_4\}$, achieving a score
of 250. The example immediately illustrates the problem with this
reasoning: While being score-optimal, the average mass of this
solution is less than 84\si{kg} and, thus, violates the average constraint (the third
item above). Alternatively, we may find it reasonable to compute a \emph{mass-optimal}
independent set, which, in Figure~\ref{figure:ensemble} is $\{\sigma_1\}$
with a glorious score of~1, which is not the best ensemble we can find.

We solved the problem using a two-pass algorithm utilizing
0-1-linear programming. In the first pass, we computed the maximum
number of ships we could possibly add to the ensemble. To that end, we
constructed the folklore LP for maximum-mass independent sets
by introducing a \emph{binary} variable $x_i$ for every ship
 $\sigma_i\in\mathcal{S}$, maximizing $\sum_{i=1}^nx_im_i$ under the constraints $x_i+x_j\leq 1$ for every two
ships $\sigma_i,\sigma_j\in\mathcal{S}$ that share an asteroid. Let us denote
this LP with $\Pi$. See the right column in
Figure~\ref{figure:ensemble}. To find the maximum number of
simultaneously usable ships, we use the algorithm displayed in Listing~\ref{algo:nships}.

\begin{lstlisting}[
  style=pseudocode,
  backgroundcolor=,
  label=algo:nships,
  float=htb,
  caption={An 0-1-LP-based algorithm to find the maximum
    number of ships in $\mathcal{S}$ that can be added to an ensemble
    simultaneously. The algorithm is based on \emph{row generation}, i.\,e., a sequence of similar LPs is generated by adding additional constraints (in Line~\ref{algo:nships:incremental}) and incrementally solved (in Line~\ref{algo:nships:optimal2}).}
]
algorithm $\textsf{maximum-number-of-ships}(\mathcal{S})$
  $\Pi$ <- $\text{init LP from }\mathcal{S}$
  $\vec x$ <- $\text{optimal solution of }\Pi$ $\label{algo:nships:optimal1}$ // mass-optimal solution ...
  $k$ <- $\sum_{i=1}^nx_i$ // ... with $\commented{k}$ ships
  while $\frac{1}{k}\sum_{i=1}^nx_im_i<\frac{k/2}{0.004}$ do // the solution violates the average-constraint $\label{algo:nships:avg}$
    $\Pi$ <- $\Pi\cup \sum_{i=1}^nx_i<k$ // search for a solution that uses less ships $\label{algo:nships:incremental}$
    $\vec x$ <- $\text{optimal solution of }\Pi$ $\label{algo:nships:optimal2}$
    $k$ <- $\sum_{i=1}^nx_i$
  return $\vec x$
\end{lstlisting}

\begin{lemma}\label{lemma:nships}
  If the algorithm from Listing~\ref{algo:nships} on input $\mathcal{S}$ outputs a solution $\vec x$
  with $k\coloneq\sum_{i=1}^nx_i$, then it is possible to construct
  from $\mathcal{S}$ an ensemble with $k$ ships, but not with $k+1$.
\end{lemma}
\begin{proof}
  The first direction is witnessed by $\vec x$ and the fact that it 
  passed the test in Line~\ref{algo:nships:avg}. For the second
  direction, consider the case that the algorithm has an intermediate
  solution $\vec x$ of $k$ ships that fails the test in
  Line~\ref{algo:nships:avg}. Then, any solution that satisfies the
  average constraint either collects more mass than
  $\sum_{i=1}^nx_im_i$ or uses fewer ships. Since $\vec x$ is
  mass-optimal by the definition of $\Pi$ and the computed solutions
  in Line~\ref{algo:nships:optimal1} or
  Line~\ref{algo:nships:optimal2}, there is no feasible solution in
  $\mathcal{S}$ collecting more mass. Hence, the only way to fix the
  average mass is by using fewer ships, which is enforced in
  Line~\ref{algo:nships:incremental} by adding the cardinality
  constraint $\sum_{i=1}^nx_i<k$.
\end{proof}

Once we know the maximum number $k$ of ships, we compute a score-optimal
independent set. Note that by Lemma~\ref{lemma:nships} we do not need to
consider ensembles with more than $k$ ships. However, a score-optimal
independent set may, of course, use fewer ships. The strategy is very similar:
We construct an LP $\Pi'$ that has the same variables and constraints
as $\Pi$, and additionally has the cardinality constraint $\sum_{i=1}^nx_i\leq k$
(for the $k$ we computed with Listing~\ref{algo:nships}), and which has to
maximize the objective function $\sum_{i=1}^nx_is_i$ (note that we replaced
``$m_i$'' by ``$s_i$''). We apply row generation again to find the global
optimum; see Listing~\ref{algo:ensemble} for details. The algorithm differs from
Listing~\ref{algo:nships} only in using the score-optimal LP, and by
excluding only a single solution in lines~\ref{algo:ensemble:index}
and~\ref{algo:ensemble:incremental} (instead of all solutions of size $k$).

\begin{lstlisting}[
  style=pseudocode,
  backgroundcolor=,
  label=algo:ensemble,
  float=htpb,
  caption={An 0-1-LP-based algorithm that computes the score-optimal
    ensemble from a set $\mathcal{S}$ of ships.} 
]
algorithm $\textsf{optimal-ensemble}(\mathcal{S}, k)$
  $\Pi'$ <- $\text{init LP from }\mathcal{S}\text{ and }k$
  $\vec x$ <- $\text{optimal solution of }\Pi$ // score-optimal solution $\label{algo:ensemble:opt1}$
  $k$ <- $\sum_{i=1}^nx_i$
  while $\frac{1}{k}\sum_{i=1}^nx_im_i<\frac{k/2}{0.004}$ do // the solution violates the average-constraint $\label{algo:ensemble:avg}$
    $I$ <- $\{\,i\mid x_i=1\,\}$ // index set of the solution    $\label{algo:ensemble:index}$
    $\Pi'$ <- $\Pi'\cup \sum_{i\in I}x_i<k$ // exclude the solution $\label{algo:ensemble:incremental}$
    $\vec x$ <- $\text{optimal solution of }\Pi$ $\label{algo:ensemble:opt2}$
    $k$ <- $\sum_{i=1}^nx_i$
  return $\vec x$
\end{lstlisting}

\begin{theorem}\label{theorem:ensemble}
  On input of $\mathcal{S}$ and $k$, the algorithm from
  Listing~\ref{algo:ensemble} outputs the ensemble of the highest
  possible score encoded as binary vector $\vec x$.
\end{theorem}
\begin{proof}
  By Lemma~\ref{lemma:nships}, it is safe to work with $\Pi'$. Clearly,
  the output of the algorithm of Listing~\ref{algo:ensemble} is a
  solution since it is a solution of $\Pi'$ and since it passes the
  average constraint in Line~\ref{algo:ensemble:avg}. To see that the
  solution is optimal, first observe that the initial $\vec x$
  computed in Line~\ref{algo:ensemble:opt1} is an upper bound on the
  solution since it is score-optimal. However, it may not be a feasible
  solution since it may violate the average constraint. Note that
  the best solution that satisfies the average constraint may use the
  same amount of ships as $\vec x$ since $\Pi'$ only
  optimizes for the score.

  However, if $\vec x$ fails the test in Line~\ref{algo:ensemble:avg},
  we can conclude that no solution is a \emph{superset}
  of (the set described by) $\vec x$. That is because every ship has a
  positive score and, hence, adding more ships
  would increase the score, contradicting the fact that
  $\vec x$ was already score-optimal. Hence, by defining $k\coloneq\sum_{i=1}^nx_i$ and $I\coloneq\{\,i\mid
  x_i=1\,\}$, and by adding the constraint $\sum_{i\in I}x_i<k$ to
  $\Pi'$, we exclude exactly one solution, namely $\vec x$. Repeatedly
  applying the argument to the solution computed in
  Line~\ref{algo:ensemble:opt2} establishes the claim of the theorem.
\end{proof}

We remark that we could immediately compute an ensemble using the
algorithm from Listing~\ref{algo:ensemble} (without using Listing~\ref{algo:nships}), i.\,e., by setting
$k=n$. However, there may be exponentially many solutions to $\Pi'$, while
Listing~\ref{algo:nships} terminates after at most $n$ rounds. In
general, it is, thus, notably faster to run both algorithms in sequence,
as described above. During the competition, we solved the integer
linear programs in lines~\ref{algo:nships:optimal1}
and~\ref{algo:nships:optimal2} in Listing~\ref{algo:nships}, and in
lines~\ref{algo:ensemble:opt1} and~\ref{algo:ensemble:opt2} in
Listing~\ref{algo:ensemble} using the open-source solver
SCIP~\cite{BestuzhevaEA23}.

\subsection{Automatic Removal of Asteroids}\label{sec:splicing}

Some of the ships $(m,s,A)$ computed in Section~\ref{sec:phoenix} had the
unfortunate property that they harvested just the bare minimum of mass from an
asteroid $a\in A$ that was heavily penalized by the bonus function. This
property is deplorable if \emph{(i)} the ship is in general good ($s$ and $m$
are large), and \emph{(ii)} there is another good ship $(m',s',A')$ with $A\cap
A'=\{a\}$. In this case, we would eventually like to add both ships to our final solution while removing $a$ from one of the trajectories and gluing the, now broken, branches. Note that it is
generally not easy to automatically change a ship's trajectory to include other
asteroids, but removing one from the trajectory is trivial. We can adapt 
$\Pi'$ to find removable asteroids as follows: We keep the indicator
variables $x_i$ for every ship $\sigma_i\in\mathcal{S}$ and introduce additional
binary variables $x_i^a$ for every asteroid $a\in A_i$ (with the semantic
``remove asteroid $a$ from ship $\sigma_i$''). For every pair of ships
$\sigma_i,\sigma_j\in\mathcal{S}$ with $i\neq j$ and $A_i\cap A_j\neq\emptyset$
we replace the constraints $x_i+x_j\leq 1$ by the following constraints for
every $a\in A_i\cap A_j$:
\[
  -x_i + -x_j + x_i^a + x_j^a\geq -1.
\]
This constraint can be read as: ``Either ship $\sigma_i$ or ship $\sigma_j$ is not in the
ensemble, or asteroid $a$ is removed from at least one of the ships''. Of
course, if we remove an asteroid, the score obtained by a ship is reduced. Let
$r_i^a$ denote the score that ship $\sigma_i$ obtains from harvesting asteroid $a$,
then the new objective function is:
\[
  \mathrm{maximize}\, \sum_{i=1}^n x_i s_i - \sum_{i=1}^n\sum_{a\in A_i}x_i^a r_i^a.
\]
We can update Listing~\ref{algo:ensemble} by replacing $\Pi'$ with the
LP described above. The output will then be a score-optimal ensemble
(described by the values of the $x_i$) together with the precise information
on which asteroids need to be removed from which ships (encoded by the $x_i^a$).

\section{The Final Assembled Solution}
\label{sec:final_assembled_solution}
Figure~\ref{fig:nabu.png}, kindly provided by the competition organizers \cite{gtoc12} gives an overview of our team's final solution found applying the selection procedure described in Sect.~\ref{sec:patch} to the pool of 371 ships able to return more than 580\si{kg} from the various runs of the search algorithms outlined in Sect.~\ref{sec:oneship} and a local refinement step applied to adjust the various epochs as to maximize the overall collected mass of each self-contained ship. The solution consists of 28 ships collectively mining $18\,475\,\mathrm{kg}$ of material from 249 asteroids, resulting in a total score of $15\,728$ (including bonus coefficients). The ship selection algorithm, when ignoring the bonus coefficient, returns a different selection corresponding to a slightly better solution at $18\,562\,\mathrm{kg}$, signaling the robustness and diversity of our final ship pool to the solutions found by other teams during the competition.

\begin{figure}[htb]
  \centering  \includegraphics[width=\textwidth,clip=true, trim=16 26 38 25]{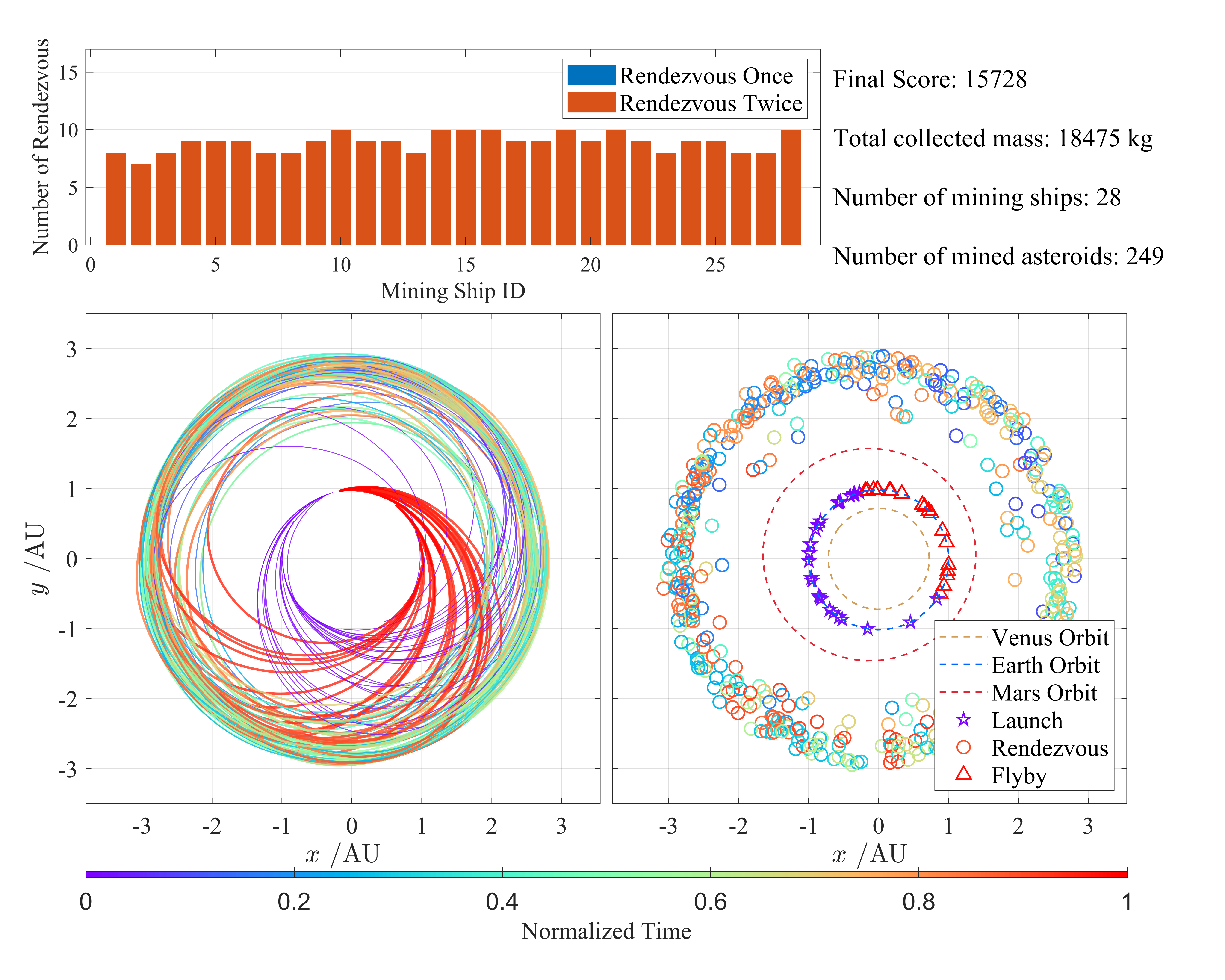}%,scale=.55, clip=true, trim=0 160 375 0
  \caption{Top: Number of rendezvous with an asteroid per ship. Since our solution is comprised solely of self-contained ships, each asteroid is rendezvoused twice by the same ship. Bottom: Projection of the trajectories of our solution (left) and of events (right) into the ecliptic. The temporal dimension is colour coded. Source: \url{https://github.com/GTOC12-official/}.}
\label{fig:nabu.png}
\end{figure}

As mentioned previously, a consequence of our search procedure (described in Sect.~\ref{sec:oneship}) is that all ships are self-contained, in that each of them collects material from the same asteroids on which they deployed miners. Each ship performs rendezvous with seven to ten asteroids twice during the mission. The bird's-eye view of the trajectories in the $xy$ plane in Fig.~\ref{fig:nabu.png} showcases that most mining activity occurs in the high asteroid density region between $2.7$ and $3\,\mathrm{AU}$ (see Fig.~\ref{fig:asteroids_3d_density}). Some of the outliers in this respect correspond to ships resulting from targeted searches on thin elliptical or inclined torus subsets of $\mathcal A$ (Sect.~\ref{sec:patch}). Of note is that none of the 28 selected ships made use of a Mars or Venus gravity assist.

Table~\ref{table:nabu} lists for each of the 28 ships constituting our solution their contribution to the final score, the collected mass, the set of visited asteroids, the search algorithm that produced them, and the associated figure.
\begin{table}[htbp]
\caption{Properties of the 28 ships of our solution: the score, the collected mass, the number of visited asteroids, and the source algorithm that produced them. ‘Scaffolding’ and ‘beam search’ refer to the algorithms of Sects.~\ref{sec:scaffolding} and~\ref{sec:phoenix}, respectively. An asterisk marks targeted searches on small subsets of $\mathcal A$, as discussed in Sect.~\ref{sec:patch}.}
\label{table:nabu}
\centering
\setlength{\tabcolsep}{12pt} 
\renewcommand{\arraystretch}{1.0}
\begin{tabular}{rcccll}
 \emph{Ship} & \emph{Score} & \emph{Mass} & \emph{Asteroids} & \emph{Algorithm} & \emph{Figures} \\[1.5ex]
% \hline
 1 & 542 & 649 kg & 8 & scaffolding \\
 2 & 549 & 613 kg & 7 & scaffolding \\
 3 & 565 & 659 kg & 8 & beam search \\
 4 & 554 & 687 kg & 9 & beam search \\
 5 & 586 & 586 kg & 9 & beam search* \\
 6 & 573 & 698 kg & 9 & beam search \\
 7 & 543 & 653 kg & 8 & scaffolding \\
 8 & 595 & 610 kg & 8 & beam search* \\
 9 & 587 & 699 kg & 9 & beam search & Figure~\ref{fig:r_of_t_plots} middle \\
 10 & 537 & 659 kg & 10 & beam search \\
 11 & 541 & 645 kg & 9 & beam search \\
 12 & 562 & 695 kg & 9 & beam search \\
 13 & 538 & 653 kg & 8 & beam search \\
 14 & 565 & 670 kg & 10 & beam search \\
 15 & 560 & 686 kg & 10 & beam search \\
 16 & 564 & 643 kg & 10 & beam search \\
 17 & 567 & 661 kg & 9 & beam search \\
 18 & 567 & 661 kg & 9 & beam search \\
 19 & 547 & 655 kg & 10 & beam search \\
 20 & 568 & 683 kg & 9 & beam search \\
 21 & 570 & 697 kg & 10 & beam search \\
 22 & 618 & 618 kg & 9 & beam search* & Figure~\ref{fig:r_of_t_plots} top \\
 23 & 556 & 682 kg & 8 & scaffolding & Figures~\ref{fig:scaffold_example} and~\ref{fig:r_of_t_plots} bottom \\
 24 & 565 & 660 kg & 9 & beam search \\
 25 & 559 & 666 kg & 9 & beam search \\
 26 & 555 & 641 kg & 8 & scaffolding \\
 27 & 545 & 653 kg & 8 & beam search \\
 28 & 555 & 693 kg & 10 & beam search
\end{tabular}
%\tablefoot{.}
\end{table}
In total five of the ships selected were produced by the scaffolding algorithm (Sect.~\ref{sec:scaffolding}) and 23 by the beam search (Sect.~\ref{sec:phoenix}). This discrepancy should however not be interpreted in terms of algorithm quality, since we also launched a disproportionately larger amount of searches with the beam search.

In the top panel of Fig.~\ref{fig:r_of_t_plots} we take a closer look at our ship with the highest score of 618 by means of an $r(t)$ plot.
\begin{figure}
  \centering  \includegraphics[width=0.85\textwidth]{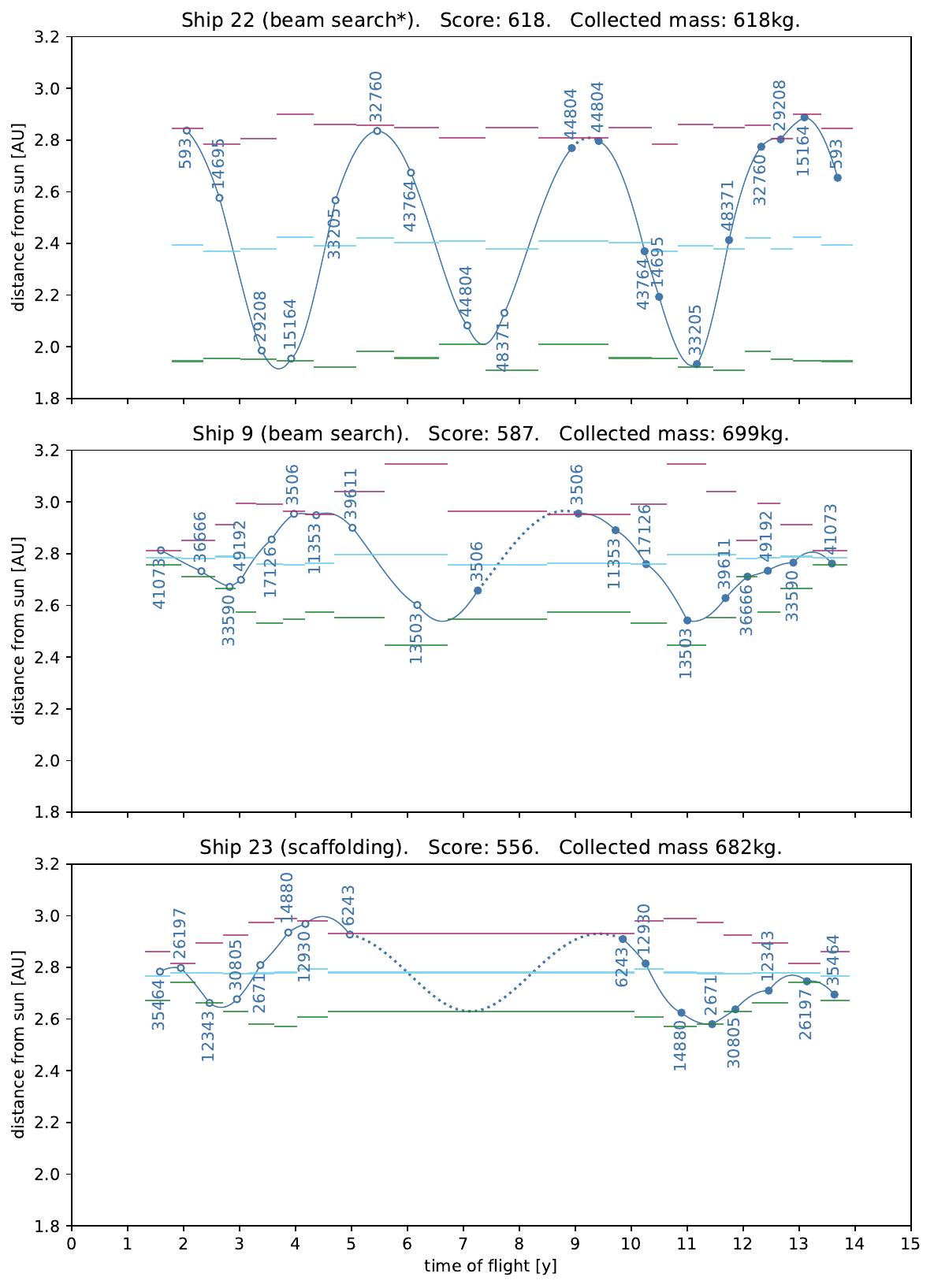}%,scale=.55, clip=true, trim=0 160 375 0
  \caption{Plots of $r(t)$ of the highest scoring ship (top) and of the highest mass returning ships produced by the beam search (middle) and scaffolding algorithm (bottom).
  Full/empty circles indicate deployment/pickup rendezvous with the labeled asteroids. Transfers are plotted solid, coastings on an asteroid dotted. The semi-major axis, the pericentre distance, and the apocentre distance of the rendezvoused asteroids are indicated by the blue, green, and red bars respectively. The chosen range on the distance axis is common to all three plots to ease a direct comparison of the geometries of the trajectories.}
\label{fig:r_of_t_plots}
\end{figure}
It is a product of a beam search on a thin elliptical torus subset of $\mathcal A$ around the orbit of the initially rendezvoused asteroid. Due to this restricted search space, all rendezvoused asteroids are on orbits with similar characteristics, and the ship finds smooth transitions between them. As discussed at the beginning of Sect.~\ref{sec:patch}, the search space restriction was an ad-hoc attempt to limit the impact of bonus coefficients by forcing the beam search to explore trajectories which it would otherwise rate unfavourably (in this case elliptical ones) and indeed the collected mass for this ship translates unaffected into score. It is important to stress however, that such ships could not have made up the bulk of our solution, since their collected mass is low. Higher mass collecting ships were required to allow for a 28 ship solution and avoid bottlenecking.

The middle panel of Fig.~\ref{fig:r_of_t_plots} shows the $r(t)$ plot of the ship of our solution with the highest mass return of $699\,\mathrm{kg}$, which is a product of the beam search. Here the unconstrained search space allowed the algorithm to find and efficiently mine the high density region just below $3\,\mathrm{AU}$ (Fig.~\ref{fig:asteroids_3d_density}). Apart from outliers, the rendezvous cluster in the beginning (deployment) and end phases (collection) of the mission are interpreted as consequences of the look-ahead heuristics of the beam search.

Finally, the bottom panel of Fig.~\ref{fig:r_of_t_plots} shows the $r(t)$ plot of the highest mass returning ship produced by the scaffolding algorithm. Also, this ship was attracted to efficiently mine in the high density region. By construction, the asteroids are rendezvoused approximately symmetrically around a central period in which the ship is coasting on an asteroid. The algorithm seeks to maximise this coasting period, and to achieve this, the ship invests in relatively expensive thrusts.

What all the above examples have in common is a low fluctuation in semi-major axis, and hence also in period, of the rendezvoused asteroids. This is in part a consequence of the energy-efficient transfers found by the algorithms as well as of the fact that similar periods ease the opportunities for revisits, in particular for near-circular trajectories.

The two highest mass returning ships that were found during the competition collected $704\,\mathrm{kg}$ of ore. These ships, however, were not selected for the submitted solution due to the presence of highly desirable and targeted asteroids in them, resulting in a harsh penalty induced by bonus coefficients.

\section{Conclusions}
The low-thrust multiple-rendezvous trajectories designed for the asteroid mining mission featured in the 12th edition of the Global Optimization Competitions can be well approximated by a first-order correction to their purely ballistic representation. This approximation offers significant efficiency advantages during the exploration of the vast combinatorial search space. It enables the avoidance of solving complex optimal control problems and instead facilitates the use of Lambert's arc computations for the trajectory analysis. Furthermore, we found that ignoring the collaborative aspect of the mining mission and focusing on the design of self-sufficient ships, mining and collecting all resources, allows to find ships returning $> 700$ kg of mined resources, albeit not consistently.

\section{Acknowledgment}
The authors wish to extend our sincere gratitude to Carlo Cavazzoni, Alessandro Russo, Sciarappa Antonio, and Franco Ongaro for their assistance in facilitating access to the Leonardo computing infrastructure in the weeks leading up to the competition. Although it turned out to be unnecessary for the particular strategy we devised, we nevertheless appreciated the opportunity and support provided.

\bibliographystyle{unsrt}
\bibliography{main}

\end{document}